\documentclass[runningheads]{llncs}

\usepackage{url}            
\usepackage{booktabs}       
\usepackage{amsfonts}       
\usepackage{nicefrac}       
\usepackage{microtype}      

\usepackage{amssymb}
\usepackage{amsmath}
\usepackage{bbm}
\usepackage{caption}
\usepackage{color}
\usepackage{enumerate}
\usepackage{enumitem}
\usepackage{float}
\usepackage{url}
\usepackage{graphicx}
\usepackage{hyperref}
\usepackage{stmaryrd}
\usepackage{stmaryrd}
\usepackage{subcaption}
\usepackage{xspace}
\usepackage{float}
\usepackage[numbers]{natbib}
\usepackage{algorithm}
\usepackage[noend]{algorithmic}
\usepackage{wrapfig}
\usepackage{tikz}

\begin{document}

\title{A Framework for Transforming Specifications in Reinforcement Learning}
\titlerunning{Transforming Specifications in Reinforcement Learning}
%
\author{Rajeev Alur\and
Suguman Bansal\and
Osbert Bastani\and
Kishor Jothimurugan}
\authorrunning{R. Alur et al.}
%
\institute{University of Pennsylvania}
\maketitle              

\spnewtheorem{open}{Open Problem}{\bfseries}{\itshape}

\newcommand{\x}{\mathbf{x}}
\newcommand{\s}{\mathbf{s}}
\newcommand{\z}{\mathbf{z}}
\newcommand{\y}{\mathbf{y}}
\renewcommand{\r}{\mathbf{r}}

\newcommand{\code}[1]{{\texttt {#1}}}

\newcommand{\R}{\mathbb{R}}
\newcommand{\E}{\mathbb{E}}
\newcommand{\C}{\mathbb{C}}
\newcommand{\D}{\mathcal{D}}
\newcommand{\J}{\mathcal{J}}
\renewcommand{\P}{\mathcal{P}}
\renewcommand{\L}{\mathcal{L}}
\newcommand{\F}{\mathcal{F}}
\renewcommand{\S}{\mathbb{S}}
\newcommand{\M}{\mathcal{M}}
\newcommand{\A}{\mathcal{A}}
\newcommand{\G}{\mathcal{G}}
\newcommand{\Rc}{\mathcal{R}}
\newcommand{\N}{\mathbb{N}}
\newcommand{\p}{\phi}
\newcommand{\pv}{\varphi}
\newcommand{\traj}{\code{Runs}}
\newcommand{\safe}{\code{safe}}
\newcommand{\opt}{\code{opt}}
\newcommand{\ltl}{\textsc{ltl}}
\newcommand{\term}{\text{term}}
\newcommand{\choice}[2]{#1 \; \code{or} \; #2}
\newcommand{\reach}{\code{reach}}
\newcommand{\avoid}[1]{\code{avoid} \; #1}
\newcommand{\dist}[2]{\code{dist}(#1,#2)}
\newcommand{\semantics}[1]{{\llbracket #1 \rrbracket}}
\newcommand{\norm}[1]{{\lVert #1 \rVert}}
\newcommand{\eventually}[1]{\code{achieve} \; #1}
\newcommand{\always}[1]{~ \code{ensuring} \; #1}
\newcommand{\true}{\code{true}}
\newcommand{\false}{\code{false}}

\newcommand{\toolname}{\textsc{Spectrl}\xspace}
\newcommand{\tltl}{\textsc{Tltl}\xspace}
\newcommand{\dirl}{\textsc{DiRL}\xspace}
\newcommand{\dirlnospace}{\textsc{DiRL}}

\newcommand{\prop}{\mathcal{P}}
\newcommand{\last}{\code{last}}
\newcommand{\avg}{\text{avg}}

\newcommand{\osbert}[1]{{\color{red}{#1}}}
\renewcommand{\sb}[1]{{\color{cyan}{#1}}}
\newcommand{\jk}[1]{{\color{blue}{#1}}}

\begin{abstract}
Reactive synthesis algorithms allow automatic construction of policies to control
an environment modeled as a Markov Decision Process (MDP) that are optimal with respect to high-level
temporal logic specifications. However, they assume that the MDP model is known a priori.
Reinforcement Learning (RL) algorithms, in contrast, are designed to learn an optimal policy when
the transition probabilities of the MDP are unknown, but require the user to associate 
local rewards with transitions.
The appeal of high-level temporal logic specifications has motivated
research to develop RL algorithms for synthesis of policies from  specifications.
To understand the techniques, and nuanced variations in their theoretical guarantees, in the growing
body of resulting literature, we develop a formal framework for defining 
transformations among RL tasks with different forms of objectives. 
We define the notion of a sampling-based reduction to transform a given MDP into another one which can be simulated even when the transition probabilities of the original MDP are unknown.
We formalize the notions of preservation of optimal policies, convergence, and robustness of such reductions.
We then use our framework to restate known results, establish new results to fill in some gaps,
and identify open problems. In particular, we show that certain kinds of reductions from LTL specifications to reward-based ones do not exist, and prove the non-existence of RL algorithms with PAC-MDP guarantees for safety specifications.

\end{abstract}

\section{Introduction}
\label{Sec:Intro}

In reactive synthesis for probabilistic systems, the system is typically modeled
as a (finite-state) Markov Decision Process (MDP), and the desired behavior of the system is given
as a logical specification.
For example, in robot motion planning, the model captures the physical environment in which the robot is operating and how the robot updates its position in response to the available control commands;
and the logical requirement can specify that the robot should always avoid obstacles and eventually reach 
all of the specified targets.
The synthesis algorithm then needs to compute a control policy that maximizes
the probability that an infinite execution of the system under the policy satisfies the logical specification.
There is a well developed theory of reactive synthesis for MDPs with respect to temporal logic specifications, accompanied by 
tools optimized with heuristics for improving scalability
and practical applications (for instance, see \cite{BaierAFK18} for a survey).

Reactive synthesis algorithms assume that the transition probabilities in the MDP modeling the environment
are known a priori. In many practical settings, these probabilities are not known, and 
the model needs to be learnt by exploration.
Reinforcement learning (RL) has emerged to be an effective paradigm for synthesis of control
policies in this scenario.
The optimization criterion for policy synthesis using RL is typically specified
by associating a local reward with each transition and aggregating the sequence of local rewards 
along an infinite execution using discounted-sum or limit-average operators.
Since an RL algorithm is learning from samples, it is expected to compute a sequence of
approximations to the optimal policy with guaranteed convergence.
Furthermore, ideally, the algorithm should have a PAC (Probably Approximately Correct) guarantee
regarding the number of samples needed to ensure that
the value of the policy it computes is within a specified bound of that of an optimal policy
with a probability greater than a specified threshold.
RL algorithms with convergence and efficient PAC guarantees are known for discounted-sum rewards \cite{watkins1992q,strehl2006pac}.


A key shortcoming of RL algorithms is that the user must manually encode the desired behavior by associating
rewards with system transitions.
An appealing alternative is to instead have the user provide a high-level logical specification encoding the task.
First, it is more natural to specify the desired properties of the global behavior, such
as ``always avoid obstacles and reach targets in this specified order'', in 
a logical formalism such as temporal logic.
Second, logical specifications facilitate testing and verifiability since it can
be checked independently whether the synthesized policy satisfies the logical requirement.
Finally, we can assume that the learning algorithm knows the logical specification
in advance, unlike the local rewards learnt during model exploration, thereby 
opening up the possibility of design of specification-aware learning algorithms.
For example, the structure of a temporal-logic formula can be exploited 
for hierarchical and compositional learning to reduce sample complexity in practice~\cite{icarte2020reward,jothimurugan2021compositional}.

This has motivated many researchers to design RL algorithms for logical specifications~\cite{aksaray2016q, brafman2018ltlf,de2019foundations, hasanbeig2018logically, littman2017environmentindependent, hasanbeig2019, yuan2019modular, moritz2019, ijcai2019-0557, jiang2020temporallogicbased,li2017reinforcement,icarte2018using,jothimurugan2022specification}.
The natural approach is to (1) translate the logical specification to an automaton
that accepts executions that satisfy the specification,
(2) define an MDP that is the product of the MDP being controlled and the specification
automaton, (3) associate rewards with the transitions of the product MDP so that
either discounted-sum or limit-average aggregation (roughly) captures
acceptance by the automaton, and (4) apply an off-the-shelf RL algorithm such as Q-learning to
synthesize the optimal policy.
While this approach is typical to most papers in this rapidly growing body of research,
and many of the proposed techniques have been shown to work well empirically,
there are many nuanced variations in terms of their theoretical guarantees.
Instead of attempting to systematically survey the literature, we classify it in following broad categories:
some only consider finite executions with a known time horizon \cite{aksaray2016q};
some provide convergence guarantees only when the optimal policy satisfies the specification almost surely; some provide PAC guarantees only when certain properties regarding the transition probabilities of the MDP are known~\cite{daca2017faster, ashok2019pac, fu2014probably};
some include a parameterized reduction \cite{bozkurt2020control,moritz2019}---the parameter being the discount factor, for instance,
establishing correctness for some value of the parameter without specifying how to compute it.
The bottom line is that there are no known RL algorithms
with convergence and/or PAC guarantees
to synthesize a policy to maximize the satisfaction of a temporal logic specification (or an impossibility result that such algorithms cannot exist).

In this paper, we propose a formal framework for defining transformations among RL tasks. 
We define an RL task to consist of an MDP ${\cal M}$ together with a specification $\phi$ for the desired policy.
The MDP is given by its states, actions, a function
to reset the state to its initial state, and a step-function that allows sampling of its transitions.
Possible forms of specifications include transition-based rewards to
be aggregated as discounted-sum or limit-average, reward machines \cite{icarte2018using}, safety, reachability,
and linear temporal logic formulas.
We then define {\em sampling-based reduction} to formalize transforming one
RL task $({\cal M}, \phi)$ to another $(\bar{\M}, \phi')$.
While the relationship between the transformed model $\bar{\M}$ and the original model ${\M}$ is
inspired by the classical definitions of simulation maps over (probabilistic) transition systems,
the main challenge is that the transition probabilities of $\bar{\M}$ cannot be directly defined
in terms of the unknown transition probabilities of ${\cal M}$. Intuitively, the
step-function to sample transitions of $\bar{\M}$ should be definable in terms 
of the step-function of ${\cal M}$ used as a black-box, and our formalization allows this.

The notion of reduction among RL tasks naturally leads to formalization of preservation of optimal policies,
convergence, and robustness (that is, policies close to optimal in one get mapped to ones close
to optimal in the other).
We use this framework to revisit existing results, fill in some gaps, and identify open problems.

{We begin with preliminaries in Section~\ref{Sec:Prelims} followed by a discussion of various kinds of specifications in Section~\ref{sec:spec_sec}. In Section~\ref{sec:reductions}, we show that it is not possible to reduce all LTL specifications to (discounted-sum) reward machines (which are reward functions with an internal state) when the underlying MDP $\M$ is kept fixed. We then define the notion of sampling-based reduction and restate existing results using our framework. In Section~\ref{sec:robust}, we introduce the notions of robust specifications and robust reductions, and show that robust sampling-based reductions do not exit for transforming safety (as well as reachability) specifications to discounted rewards. Finally, we present our result on non-existence of RL algorithms with PAC-MDP guarantees for safety (and reachability) specifications in Section~\ref{sec:no_pac}.}

\subsection*{Related Work}
{In Section~\ref{sec:existing}, we discuss some existing work on reducing LTL specifications to rewards. There is work on similar reductions for more complex lexicographic $\omega$-regular objectives \cite{hahn2021model} as well as in the context of stochastic games \cite{hahn2020model}. We discuss existing PAC learning results for logical specifications in Section~\ref{sec:pac_result}. Concurrent to our work, the authors of \cite{yang2021reinforcement} show that PAC-MDP algorithms do not exist for any \emph{non-finitary} LTL objective.

Closely related to this work is the work on expressivity of discounted rewards \cite{abel2021expressivity} which studies whether certain kinds of tasks can be encoded using discounted rewards. There are a couple of key differences to our work. First, they do not consider reductions that involve modifying the underlying MDP $\M$. Second, the tasks considered are based on explicit orderings among policies or trajectories rather than succinct formal specifications.}
\section{Preliminaries}
\label{Sec:Prelims}

\paragraph{Markov Decision Process.}

A {Markov Decision Process} (MDP) is a tuple $\M = (S, A, s_0, P)$, where $S$ is a finite set of states,
$s_0$ is the initial state,\footnote{A distribution $\eta$ over initial states can be modeled by adding a new state $s_0$ from which taking any action leads to a state sampled from $\eta$.} $A$ is a finite set of actions, and $P:S \times A \times S \to [0,1]$ is the transition probability function, with $\sum_{s'\in S}P(s,a,s') = 1$ for all $s\in S$ and $a\in A$.


An \emph{infinite run} $\zeta \in (S\times A)^{\omega}$ is a sequence $\zeta = s_0{a_0}s_1{a_1}\ldots$, where $s_i \in S$ and $a_i\in A$ for all $i\in\N$. Similarly, a \emph{finite run} $\zeta\in(S\times A)^*\times S$ is a finite sequence $\zeta = s_0{a_0}s_1{a_1}\ldots a_{t-1}s_t$. For any run $\zeta$ of length at least $j$ and any $i\leq j$, we let
$\zeta_{i:j}$ denote the subsequence $s_i{a_i}s_{i+1}{a_{i+1}}\ldots{a_{j-1}}s_j$. We use $\traj(S,A) = (S\times A)^{\omega}$ and $\traj_f(S,A) = (S\times A)^*\times S$ to denote the set of infinite and finite runs, respectively. 

Let $\D(A) = \{\Delta:A\to[0,1]\mid\sum_{a\in A}\Delta(a) = 1\}$ denote the set of all distributions over actions. A policy $\pi:\traj_f(S,A)\to\D(A)$ maps a finite run $\zeta\in\traj_f(S,A)$ to a distribution $\pi(\zeta)$ over actions. We denote by $\Pi(S,A)$ the set of all such policies. A policy $\pi$ is \emph{positional} if $\pi(\zeta) = \pi(\zeta')$ for all $\zeta, \zeta'\in\traj_f(S,A)$ with $\last(\zeta) = \last(\zeta')$ where $\last(\zeta)$ denotes the last state in the run $\zeta$. A policy $\pi$ is deterministic if, for all finite runs $\zeta\in\traj_f(S,A)$, there is an action $a\in A$ with $\pi(\zeta)(a) = 1$.

Given a finite run $\zeta=s_0a_0\ldots a_{t-1}s_t$, the \emph{cylinder} of $\zeta$, denoted by $\code{Cyl}(\zeta)$, is the set of all infinite runs starting with prefix $\zeta$. Given an MDP $\M$ and a policy $\pi\in\Pi(S,A)$, we define the probability of the cylinder set by $\D^{\M}_{\pi}(\code{Cyl}(\zeta)) = \prod_{i=0}^{t-1}\pi(\zeta_{0:i})(a_i)P(s_i, a_i, s_{i+1})$. It is known that $\D_{\pi}^{\M}$ can be uniquely extended to a probability measure over the $\sigma$-algebra generated by all cylinder sets.

\paragraph{Simulator.} 

In reinforcement learning, the standard assumption is that the set of states $S$, the set of actions $A$, and the initial state $s_0$ are known but the transition probability function $P$ is unknown. The learning algorithm has access to a simulator $\S$ which can be used to sample runs of the system $\zeta\sim\D_{\pi}^{\M}$ using any policy $\pi$. The simulator can also be the real system, such as a robot, that $\M$ represents. Internally, the simulator stores the current state of the MDP which is denoted by $\S.\code{state}$. It makes the following functions available to the learning algorithm.

\begin{description}
\item[\rm $\S.\code{reset()}$:] This function sets $\S.\code{state}$ to the initial state $s_0$.

\item[\rm $\S.\code{step(}a\code{)}$:] Given as input an action $a$, this function samples a state $s'\in S$ according to the transition probability function $P$---i.e., the probability that a state $s'$ is sampled is $P(s,a,s')$ where $s=\S.\code{state}$. It then updates $\S.\code{state}$ to the newly sampled state $s'$ and returns $s'$.


\end{description}
Simulation models without the \code{reset()} function have also been studied~\cite{kearns2002near, strehl2006pac}. In this paper, we allow resets, however, we believe that our results also apply to settings in which resets are not allowed.



\section{Task Specification}\label{sec:spec_sec}
In this section, we present many different ways in which one can specify the objective of the learning algorithm. We define a \emph{reinforcement learning task} to be a pair ($\M$, $\p$) where $\M$ is an MDP and $\p$ is a specification for $\M$. In general, a specification $\p$ for $\M = (S,A,s_0,P)$ defines a function $J^{\M}_{\p}:\Pi(S,A)\to\R$ and the reinforcement learning objective is to compute a policy $\pi$ that maximizes $J^{\M}_{\p}(\pi)$. Let $\J^*(\M,\p) = \sup_{\pi}J^{\M}_{\p}(\pi)$ denote the maximum value of $J^{\M}_{\p}$. We let $\Pi_{\opt}(\M,\p)$ denote the set of all optimal policies in $\M$ w.r.t. $\p$---i.e., $\Pi_{\opt}(\M,\p) = \{\pi\mid J^{\M}_{\p}(\pi)=\J^*(\M,\p)\}$. In many cases, it is sufficient to compute an $\varepsilon$-optimal policy $\tilde{\pi}$ with $J_{\p}^{\M}(\tilde{\pi})\geq \J^*(\M,\p)-\varepsilon$; we let $\Pi_{\opt}^{\varepsilon}(\M,\p)$ denote the set of all $\varepsilon$-optimal policies in $\M$ w.r.t. $\p$.  

\subsection{Rewards}

The most common kind of specifications used in reinforcement learning is reward functions that map transitions in $\M$ to real values. {We first define the more general \emph{reward machines} and then define standard transition-based reward functions as a special case.}


\paragraph{Reward Machines.} Reward Machines \cite{icarte2018using} extend simple transition-based reward functions to history-dependent ones by using an automaton model. Formally, a reward machine for an MDP $\M = (S,A,s_0,P)$ is a tuple $\Rc = (U, u_0, \delta_u, \delta_r)$, where $U$ is a finite set of states, $u_0$ is the initial state, $\delta_u: U\times S\to U$ is the state transition function, and $\delta_r: U\to[S\times A\times S\to\R]$ is the reward function. Given an infinite run $\zeta=s_0{a_0}s_1{a_1}\ldots$, we can construct an infinite sequence of reward machine states $\rho_{\Rc}(\zeta) = u_0u_1,\ldots$ defined by $u_{i+1} = \delta_u(u_i, s_{i+1})$. Then, we can assign either a discounted-sum or a limit-average reward to $\zeta$:
\begin{itemize}
    \item \emph{Discounted Sum.} Given a discount factor $\gamma\in]0,1[$, the full specification is $\p = (\Rc, \gamma)$ and we have $$\Rc_{\gamma}(\zeta) = \sum_{i=0}^{\infty}\gamma^i \delta_r(u_i)(s_i, a_i, s_{i+1}).$$
    Though less standard, one can use different discount factors in different states of $\M$, in which case we have $\gamma:S\to ]0,1[$ and
    $$\Rc_{\gamma}(\zeta) = \sum_{i=0}^{\infty}\Big(\prod_{j=0}^{i-1}\gamma(s_j)\Big) \delta_r(u_i)(s_i, a_i, s_{i+1}).$$
    The value of a policy $\pi$ is $J^{\M}_{\p}(\pi) = \E_{\zeta\sim\D^{\M}_{\pi}}[\Rc_{\gamma}(\zeta)]$. 
    \item \emph{Limit Average.} The specification is just a reward machine $\p = \Rc$. The $t$-step average reward of the run $\zeta$ is  $$\Rc_{\avg}^t(\zeta) = \frac{1}{t}\sum_{i=0}^{t-1}\delta_r(u_i)(s_i, a_i, s_{i+1}).$$
    The value of a policy $\pi$ is $J^{\M}_{\p}(\pi) = \liminf_{t\to\infty}\E_{\zeta\sim\D^{\M}_{\pi}}[\Rc_{\avg}^t(\zeta)]$.
\end{itemize}
{A standard transition-based reward function $R$ is simply a reward machine $\Rc$ with a single state $u_0$; in this case, we use $R(s, a, s')$ to denote $\delta_r(u_0)(s,a,s')$.}

\subsection{Abstract Specifications}
The above specifications are defined w.r.t. a given set of states $S$ and actions $A$, and can only be interpreted over MDPs with the same state and action spaces. In this section, we look at \emph{abstract specifications}, which are defined independently of $S$ and $A$. To achieve this, a common assumption is that there is a fixed set of propositions $\prop$, and the simulator provides access to a labeling function $L:S\to2^{\prop}$ denoting which propositions are true in any given state. Given a run $\zeta=s_0a_0s_1a_1\ldots$, we let $L(\zeta)$ denote the corresponding sequence of labels $L(\zeta) = L(s_0)L(s_1)\ldots$. A \emph{labeled MDP} is a tuple $\M = (S,A,s_0,P,L)$. WLOG, we only consider labeled MDPs in the rest of the paper.

\paragraph{Abstract Reward Machines.} Reward machines can be adapted to the abstract setting quite naturally. An \emph{abstract reward machine} (ARM) is similar to a reward machine except $\delta_u$ and $\delta_r$ are independent of $S$ and $A$---i.e., $\delta_u:U\times 2^{\prop}\to U$ and $\delta_r:U\to[2^{\prop}\to\R]$. Given current ARM state $u_i$ and next MDP state $s_{i+1}$, the next ARM state is given by $u_{i+1} = \delta_u(u_i, L(s_{i+1}))$, and the reward is given by $\delta_r(u_i)(L(s_{i+1}))$.

\paragraph{Languages.} Formal languages can be used to specify qualitative properties about runs of the system. A language specification $\p = \L\subseteq (2^{\prop})^{\omega}$ is a set of ``desirable" sequences of labels. The value of a policy $\pi$ is the probability of generating a sequence in $\L$---i.e.,
$$J^{\M}_{\p}(\pi) =  \D^{\M}_{\pi}\big(\{\zeta\in\traj(S,A)\mid L(\zeta)\in \L\}\big).$$
Some common ways to define languages are as follows.
\begin{itemize}
\item \emph{Reachability.} Given an accepting set of propositions $X\in 2^{\prop}$, we have $\L_{\reach}(X) = \{w\in(2^{\prop})^{\omega}\mid \exists i.\ w_i \cap X\neq\emptyset\}$.
\item \emph{Safety.} Given a safe set of propositions $X\in 2^{\prop}$, we have $\L_{\safe}(X) = \{w\in(2^{\prop})^{\omega}\mid \forall i.\ w_i \subseteq X\}$.
\item \emph{Linear Temporal Logic.} Linear Temporal Logic \cite{pnueli1977temporal} over propositions $\prop$ is defined by the grammar
$$\varphi:= b\in\prop\ |\ \varphi\lor\varphi\ |\ \lnot\varphi\ |\ \bigcirc\varphi\ |\ \varphi\ \mathcal{U}\ \varphi$$
where $\bigcirc$ denotes the ``Next" operator and $\mathcal{U}$ denotes the ``Until" operator. We refer the reader to \cite{sistla1985complexity} for more details on the semantics of LTL specifications. We use $\Diamond$ and $\Box$ to denote the derived ``Eventually" and ``Always" operators, respectively.
Given an LTL specification $\varphi$ over propositions $\prop$, we have $\L_{\ltl}(\varphi) = \{w\in(2^{\prop})^{\omega}\mid w\models \varphi\}$.
\end{itemize}

\subsection{Learning Algorithms}
\label{sec:learningalgos}

A learning algorithm $\A$ is an iterative process that in each iteration (i) either resets the simulator or takes a step in $\M$, and (ii) outputs its current estimate of an optimal policy $\pi$. A learning algorithm $\A$ induces a random sequence of output policies $\{\pi_n\}_{n=1}^{\infty}$ where $\pi_n$ is the policy output in the $n^{\text{th}}$ iteration. We consider two common kinds of learning algorithms. First, we consider algorithms that converge in the limit almost surely.

\begin{definition}
A learning algorithm $\A$ is said to converge in the limit for a class of specifications $\C$ if, for any RL task ($\M,\p$) with $\p\in\C$,
$$J_{\p}^{\M}(\pi_n) \to \J^*(\M,\p)\ \text{as}\ n\to\infty\quad \text{almost surely.}$$
\end{definition}

Q-learning \citep{watkins1992q} is an example of a learning algorithm that converges in the limit for discounted-sum rewards. There are variants of Q-learning for limit-average rewards \citep{abounadi2001learning} which have been shown to converge in the limit under some assumptions on the MDP $\M$. The second kind of algorithms is \emph{Probably Approximately Correct} (PAC-MDP) \citep{kakade2003sample} algorithms which are defined as follows.

\begin{definition}\label{def:pac-mdp}
A learning algorithm $\A$ is said to be PAC-MDP for a class of specifications $\C$ if, there is a function $h$ such that for any $p>0$, $\varepsilon>0$, and any RL task $(\M,\p)$ with $\M=(S,A,s_0,P)$ and $\p\in\C$, taking $N=h(|S|,|A|,|\p|,\frac{1}{p}, \frac{1}{\varepsilon})$, with probability at least $1-p$, we have
$$\Big|\Big\{n\mid \pi_n\notin\Pi_{\opt}^{\varepsilon}(\M,\p) \Big\}\Big|\leq N.$$
\end{definition}

We say a PAC-MDP algorithm is \emph{efficient} if the \emph{sample complexity} function $h$ is polynomial in $|S|, |A|, \frac{1}{p}$ {and} $\frac{1}{\varepsilon}$. There are efficient PAC-MDP algorithms for discounted-sum rewards \citep{kearns2002near,strehl2006pac}.

\section{Reductions}\label{sec:reductions}

There has been a lot of research on RL algorithms for reward-based specifications.
The most common approach for language-based specifications is to transform the given specification into a reward function and apply algorithms that maximize the expected reward. In such cases, it is important to ensure that maximizing the expected reward corresponds to maximizing the probability of satisfying the specification. In this section, we study such reductions and formalize a general notion of \emph{sampling-based reductions} in the RL setting---i.e., the transition probabilities are unknown and only a simulator of $\M$ is available. 

\subsection{Specification Translations}

We first consider the simplest form of reduction, which involves translating the given specification into another one. Given a specification $\p$ for MDP $\M=(S, A, s_0, P, L)$ we want to construct another specification $\p'$ such that for any $\pi\in\Pi_{\opt}(\M,\p')$, we also have $\pi\in\Pi_{\opt}(\M,\p)$. This ensures that $\p'$ can be used to compute a policy that maximizes the objective of $\p$. Note that since the transition probabilities $P$ are not known, the translation has to be independent of $P$ and furthermore the above \emph{optimality preservation} criterion must hold for all $P$.

\begin{definition}
An optimality preserving specification translation is a computable function $\F$ that maps the tuple $(S,A, s_0, L, \p)$ to a specification $\p'$ such that for all transition probability functions $P$, letting $\M = (S, A, s_0, P, L)$, we have $\Pi_{\opt}(\M,\p')\subseteq\Pi_{\opt}(\M,\p)$.
\end{definition}

A first attempt at a reinforcement learning algorithm for language-based specifications is to translate the given specification to a reward machine (either discounted-sum or limit-average). However there are some limitations to this approach. First, we show that it is not possible to reduce reachability and safety objectives to reward machines with discounted rewards.

\begin{theorem}\label{thm:discount}
Let $\prop = \{b\}$ and $\p = \L_{\reach}(\{b\})$. There exists $S$, $A$, $s_0$, $L$ such that for any discounted-sum reward machine specification $\p' = (\Rc, \gamma)$, there is a transition probability function $P$ such that for $\M = (S, A, s_0, P, L)$, we have $\Pi_{\opt}(\M,\p')\not\subseteq\Pi_{\opt}(\M,\p)$.
\end{theorem}

{The main idea behind the proof is that one can make the transition probabilities small enough so that the expected time taken to reach the goal is large while maintaining an optimal probability of 1 for eventually reaching the goal. Using this idea, it is possible to define transition probabilities such that the expected reward w.r.t. an optimal policy is smaller than the expected reward obtained by a suboptimal policy.}
\begin{figure}
\begin{center}
\begin{tikzpicture}[scale=0.14]
\tikzstyle{every node}+=[inner sep=0pt]
\draw [black] (19.1,-27) circle (3);
\draw (19.1,-27) node {$s_0$};
\draw [black] (40.6,-13.1) circle (3);
\draw (40.6,-13.1) node {$s_1$};
\draw [black] (40.6,-13.1) circle (2.4);
\draw [black] (61.7,-27) circle (3);
\draw (61.7,-27) node {$s_3$};
\draw [black] (40.6,-27) circle (3);
\draw (40.6,-27) node {$s_2$};
\draw [black] (20.544,-24.373) arc (147.65954:98.10674:24.248);
\fill [black] (37.61,-13.34) -- (36.75,-12.96) -- (36.89,-13.95);
\draw (24.8,-16.48) node [above] {$a_1\mbox{ }/\mbox{ }p_1$};
\draw [black] (59.916,-29.41) arc (-39.88125:-140.11875:25.433);
\fill [black] (59.92,-29.41) -- (59.02,-29.7) -- (59.79,-30.34);
\draw (40.4,-39.04) node [below] {$a_1\mbox{ }/\mbox{ }1-p_1$};
\draw [black] (59.258,-28.741) arc (-56.9947:-123.0053:34.62);
\fill [black] (59.26,-28.74) -- (58.31,-28.76) -- (58.86,-29.6);
\draw (40.4,-34.83) node [below] {$a_2\mbox{ }/\mbox{ }1-p_2$};
\draw [black] (43.28,-11.777) arc (144:-144:2.25);
\draw (47.85,-13.1) node [right] {$a_1\mbox{ }/\mbox{ }1$};
\fill [black] (43.28,-14.42) -- (43.63,-15.3) -- (44.22,-14.49);
\draw [black] (22.1,-27) -- (37.6,-27);
\fill [black] (37.6,-27) -- (36.8,-26.5) -- (36.8,-27.5);
\draw (29.85,-26.5) node [above] {$a_2\mbox{ }/\mbox{ }p_2$};
\draw [black] (64.38,-25.677) arc (144:-144:2.25);
\draw (68.95,-27) node [right] {$a_1\mbox{ }/\mbox{ }1$};
\fill [black] (64.38,-28.32) -- (64.73,-29.2) -- (65.32,-28.39);
\draw [black] (43.095,-25.355) arc (151.12502:-136.87498:2.25);
\draw (47.97,-26.03) node [right] {$a_1\mbox{ }/\mbox{ }p_3$};
\fill [black] (43.42,-27.98) -- (43.88,-28.8) -- (44.37,-27.93);
\draw [black] (40.6,-24) -- (40.6,-16.1);
\fill [black] (40.6,-16.1) -- (40.1,-16.9) -- (41.1,-16.9);
\draw (41.1,-20.05) node [right] {$a_1\mbox{ }/\mbox{ }1-p_3$};
\end{tikzpicture}
\end{center}

\caption{Counterexample for reducing reach specification to a discounted RM.}
\label{fig:reach_mdp}
\end{figure}
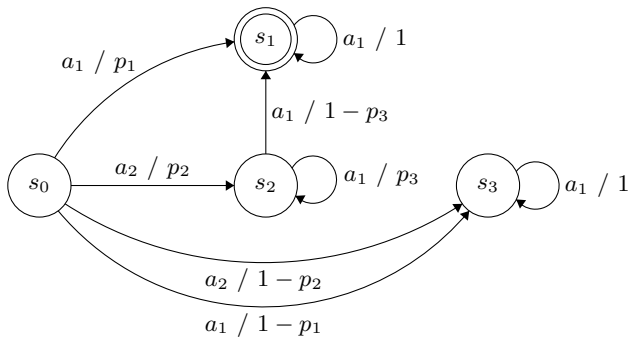

\begin{proof}
Consider the MDP in Figure~\ref{fig:reach_mdp}, which has states $S = \{s_0, s_1, s_2, s_3\}$, actions $A = \{a_1, a_2\}$, and labeling function $L$ given by $L(s_1) = \{b\}$ (marked with double circles) and $L(s_0) = L(s_2) = L(s_3) = \emptyset$. Each edge denotes a state transition and is labeled by an action followed by the transition probability; the latter are parameterized by $p_1, p_2$, and $p_3$. At states $s_1, s_2$, and $s_3$ the only action available is $a_1$.\footnote{This can be modeled by adding an additional dead state that is reached upon taking action $a_2$ in these states.} There are only two deterministic policies $\pi_1$ and $\pi_2$ in $\M$; $\pi_1$ always chooses $a_1$, whereas $\pi_2$ first chooses $a_2$ followed by $a_1$ afterwards.

For the sake of contradiction, suppose there is a $\p' = (\Rc, \gamma)$ that preserves optimality w.r.t. $\p$ for all values of $p_1$, $p_2$, and $p_3$. WLOG, we assume that the rewards are normalized---i.e. $\delta_r:U\to [S\times A\times S\to [0,1]]$. If $p_1 = p_2 = p_3 = 1$, then taking action $a_1$ in $s_0$ achieves reach probability of $1$, whereas taking action $a_2$ in $s_0$ leads to a reach probability of $0$. Hence, we must have that $\Rc_{\gamma}(s_0a_1(s_1a_1)^{\omega}) \geq \Rc_{\gamma}(s_0a_2(s_2a_1)^{\omega}) + \varepsilon$ for some $\varepsilon > 0$, as otherwise, $\pi_2$ maximizes $J_{\p'}^{\M}$ but does not maximize $J_{\p}^{\M}$.


For any finite run $\zeta\in\traj_f(S,A)$, let $\Rc_{\gamma}(\zeta)$ denote the finite discounted-sum reward of $\zeta$.
Let $t$ be such that $\frac{\gamma^{t}}{1-\gamma} \leq \frac{\varepsilon}{2}$. Then, for any $\zeta\in \traj(S,A)$, we have
\begin{align*}
\Rc_{\gamma}(s_0a_1(s_1a_1)^{\omega}) &\geq \Rc_{\gamma}(s_0a_2(s_2a_1)^{\omega}) + \varepsilon\\
&\geq \Rc_{\gamma}(s_0a_2(s_2a_1)^{t}s_2) + \varepsilon\\
&\geq \Rc_{\gamma}(s_0a_2(s_2a_1)^{t}s_2) + \frac{\gamma^t}{1-\gamma} + \frac{\varepsilon}{2}.
\end{align*}
Since $\lim_{p_3\to 1}p_3^{t} = 1$, there exists $p_3 < 1$ such that $1 - p_3^{t} \leq \frac{\varepsilon}{8}(1-\gamma)$. Let $p_1 < 1$ be such that $p_1\cdot\Rc_{\gamma}(s_0a_1(s_1a_1)^{\omega}) \geq \Rc_{\gamma}(s_0a_2(s_2a_1)^{t}s_2) + \frac{\gamma^t}{1-\gamma} + \frac{\varepsilon}{4}$ and let $p_2 = 1$. Then, we have
\begin{align*}
J_{\p'}^{\M}(\pi_1) &\geq p_1\cdot\Rc_{\gamma}(s_0a_1(s_1a_1)^{\omega})\\
&\geq \Rc_{\gamma}(s_0a_2(s_2a_1)^{t}s_2) + \frac{\gamma^t}{1-\gamma} +  \frac{\varepsilon}{4}\\
&\geq p_3^{t}\cdot\Big(\Rc_{\gamma}(s_0a_2(s_2a_1)^{t}s_2)+ \frac{\gamma^t}{1-\gamma}\Big) + \frac{\varepsilon}{4}\\
&\geq p_3^{t}\cdot\Big(\Rc_{\gamma}(s_0a_2(s_2a_1)^{t}s_2)+ \frac{\gamma^t}{1-\gamma}\Big) + (1-p_3^{t})\cdot\Big(\frac{1}{1-\gamma}\Big) +  \frac{\varepsilon}{8}\\
&> J^{\M}_{\p'}(\pi_2),
\end{align*}
where the last inequality followed from the fact that when using $\pi_2$, the system stays in state $s_2$ for at least $t$ steps with probability $p_3^t$, and the reward of such trajectories is bounded from above by $\Rc_{\gamma}(s_0a_2(s_2a_1)^{t}s_2)+ \frac{\gamma^t}{1-\gamma}$, along with the fact that the reward of all other trajectories is bounded by $\frac{1}{1-\gamma}$. This leads to a contradiction since $\pi_1$ maximizes $J_{\p'}^{\M}$ but $J_{\p}^{\M}(\pi_1) = p_1 < 1 = J_{\p}^{\M}(\pi_2)$.\qed
\end{proof}

We do not use the fact that the reward machine is finite state in our proof; therefore, the above result applies to general non-Markovian reward functions of the form $R:\traj_f(S,A)\to[0,1]$ with $\gamma$-discounted reward defined by $R_{\gamma}(\zeta) = \sum_{i=0}^{\infty}\gamma^iR(\zeta_{0:i})$. The proof can be easily modified to show the result for safety specifications as well.

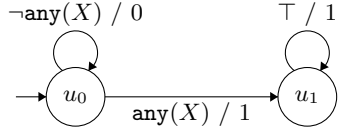
\begin{figure}
\centering
\begin{center}
\begin{tikzpicture}[scale=0.13]
\tikzstyle{every node}+=[inner sep=0pt]
\draw [black] (29.2,-27) circle (3);
\draw (29.2,-27) node {$u_0$};
\draw [black] (52.9,-27) circle (3);
\draw (52.9,-27) node {$u_1$};
\draw [black] (23.1,-27) -- (26.2,-27);
\fill [black] (26.2,-27) -- (25.4,-26.5) -- (25.4,-27.5);
\draw [black] (32.2,-27) -- (49.9,-27);
\fill [black] (49.9,-27) -- (49.1,-26.5) -- (49.1,-27.5);
\draw (41.05,-27.5) node [below] {$\code{any}(X)\mbox{ }/\mbox{ }1$};
\draw [black] (27.877,-24.32) arc (234:-54:2.25);
\draw (29.2,-19.75) node [above] {$\lnot \code{any}(X)\mbox{ }/\mbox{ }0$};
\fill [black] (30.52,-24.32) -- (31.4,-23.97) -- (30.59,-23.38);
\draw [black] (51.577,-24.32) arc (234:-54:2.25);
\draw (52.9,-19.75) node [above] {$\top\mbox{ }/\mbox{ }1$};
\fill [black] (54.22,-24.32) -- (55.1,-23.97) -- (54.29,-23.38);
\end{tikzpicture}
\end{center}
\caption{ARM for $\p= \L_{\reach}(X)$.}
\label{fig:limavgrm}
\end{figure}

The main challenge in translating to discounted-sum rewards is the fact that the rewards vanish over time and the overall reward depends primarily on the first few steps. This issue can be partly overcome by using limit-average rewards. In fact, we have the following theorem.

{\begin{theorem}
There exists an optimality preserving specification translation from reachability and safety specifications to abstract reward machines (with limit-average aggregation).
\end{theorem}

\begin{proof}
An abstract reward machine for the specification $\p = \L_{\reach}(X)$ is shown in Figure~\ref{fig:limavgrm}. Each transition is labeled by a Boolean formula over $\prop$ followed by the reward. We use $\code{any}(X)$ to denote $\bigvee_{b\in X}b$. It is easy to see that for any MDP $\M$ and any policy $\pi$ of $\M$, we have $J_{\Rc}^{\M}(\pi) = J^{\M}_{\p}(\pi)$. An ARM for $\p=\L_{\safe}(\prop\setminus X)$ is obtained by replacing the reward value $r$ by $1-r$ on all transitions.\qed
\end{proof}}
However, we can show that there does not exist an ARM for the specification $\p=\L_{\ltl}(\Box\Diamond b)$, which requires the proposition $b$ to be true infinitely often. Intuitively, the result follows from the fact that, given any ARM, we can construct an infinite word $w\in(2^{\prop})^{\omega}$ in which $b$ holds true rarely but infinitely often such that $w$ achieves a lower limit-average reward than another word $w'$ in which $b$ holds true more frequently.

\begin{theorem}\label{thm:no_arm}
Let $\prop = \{b\}$ and $\p = \L_{\ltl}(\Box\Diamond b)$. For any ARM specification $\p' = \Rc$ with limit-average rewards, there exists an MDP $\M = (S, A, s_0, P, L)$ such that $\Pi_{\opt}(\M,\p')\not\subseteq\Pi_{\opt}(\M,\p)$.
\end{theorem}

\begin{proof}
For the sake of contradiction, let $\p' = \Rc = (U, u_0, \delta_u, \delta_r)$ be an ARM that preserves optimality w.r.t $\p$ for all MDPs. The extended state transition function $\delta_u: U\times(2^{\prop})^*\to U$ is defined naturally. WLOG, we assume that all states in $\Rc$ are reachable from the initial state and that the rewards are normalized\footnote{The rewards assigned by an ARM have to be independent of $S$.}---i.e., $\delta_r:U\to [2^{\prop} \to [0,1]]$.

A \emph{cycle} in $\Rc$ is a sequence $C = u_1\ell_1u_2\ell_2\ldots \ell_{k}u_{k+1}$ where $u_i\in U$, $\ell_i\in2^{\prop}$, $u_{i+1} = \delta_u(u_i,\ell_i)$ for all $i$, $u_{k+1} = u_1$, and the states $u_1,\ldots,u_{k}$ are distinct. A cycle is \emph{negative} if $\ell_{i} = \emptyset$ for all $i$, and \emph{positive} if $\ell_{i} = \{b\}$ for all $i$. The average reward of a cycle $C$ is given by $\Rc_{\avg}(C) = \frac{1}{k}\sum_{i=1}^{k}\delta_r(u_i)(\ell_i)$. For any cycle $C = u_1\ell_1\ldots\ell_{k}u_{k+1}$ we can construct a deterministic MDP $\M_{C}$ with a single action that first generates a sequence of labels $\sigma$ such that $\delta_u(u_0, \sigma) = u_1$, and then repeatedly generates the sequence of labels $\ell_1\ldots\ell_{k}$. The limit average reward of the only policy $\pi$ in $\M_C$ is $J_{\Rc}^{\M_C}(\pi) = \Rc_{\avg}(C)$ since the cycle $C$ repeats indefinitely. 

Now, given any positive cycle $C_+$ and any negative cycle $C_{-}$, we claim that $\Rc_{\avg}(C_{+}) > \Rc_{\avg}(C_{-})$. To show this, consider an MDP $\M$ with two actions $a_1$ and $a_2$ such that taking action $a_1$ in the initial state $s_0$ leads to $\M_{C_{+}}$, and taking action $a_2$ in $s_0$ leads to $\M_{C_{-}}$. The policy $\pi_1$ that takes action $a_1$ in $s_0$ achieves a satisfaction probability of $J_{\p}^{\M}(\pi_1) = 1$, whereas the policy $\pi_2$ taking action $a_2$ in $s_0$ achieves $J_{\p}^{\M}(\pi_2) = 0$. Since $J_{\Rc}^{\M}(\pi_1) = \Rc_{\avg}(C_{+})$ and $J_{\Rc}^{\M}(\pi_2) = \Rc_{\avg}(C_{-})$, we must have that $\Rc_{\avg}(C_{+}) > \Rc_{\avg}(C_{-})$ to preserve optimality w.r.t. $\p$. Since there are only finitely many cycles in $\Rc$, there exists an $\varepsilon > 0$ such that for any positive cycle $C_{+}$ and any negative cycle $C_{-}$ we have $\Rc_{\avg}(C_{+}) \geq \Rc_{\avg}(C_{-}) + \varepsilon$.

{Consider a bottom strongly connected component (SCC) of the graph of $\Rc$. We show that this component contains a negative cycle $C_{-} = u_1\ell_1\ldots\ell_{k}u_{k+1}$ along with a second cycle $C = u_1'\ell_1'\ldots\ell_{k'}'u_{k'+1}'$ such that $u_1' = u_1$ and $\ell_1' = \{b\}$. To construct $C_-$, from any state in the bottom SCC of $\Rc$, we can follow edges labeled $\ell=\emptyset$ until we repeat a state, and let this state be $u_1$. Then, to construct $C$,  from $u_1'=u_1$, we can follow the edge labeled $\ell_1'=\{b\}$ to reach $u_2'=\delta(u_1',\ell_1')$; since we are in the bottom SCC, there exists a path from $u_2'$ back to $u_1'$.}
Now, consider a sequence of the form $C_{m} = C_{-}^mC$, where $m\in\N$. We have 
\begin{align*}
\Rc_{\avg}(C_m) &= \frac{mk\Rc_{\avg}(C_{-}) + k'\Rc_{\avg}(C)}{mk+k'} \\
&\leq \frac{mk}{mk + k'}\Rc_{\avg}(C_{-}) + \frac{k'}{mk + k'}\\
&\leq \Rc_{\avg}(C_{-}) + \frac{k'}{mk + k'}.
\end{align*}
Let $m$ be such that $\frac{k'}{mk + k'} \leq \frac{\varepsilon}{2}$ and $C_{+}$ be any positive cycle. Then, we have $\Rc_{\avg}(C_{+}) \geq \Rc_{\avg}(C_m) + \frac{\varepsilon}{2}$; therefore, there exists $p < 1$ such that $p\cdot\Rc_{\avg}(C_{+}) \geq \Rc_{\avg}(C_m) + \frac{\varepsilon}{4}$. Now, we can construct an MDP $\M$ in which (i) taking action $a_1$ in initial state $s_0$ leads to $\M_{C_{+}}$ with probability $p$, and to a dead state (where $b$ does not hold) with probability $1-p$, and (ii) taking action $a_2$ in initial state $s_0$ leads to a deterministic single-action component that forces $\Rc$ to reach $u_1$ {(recall that WLOG, all states in $\Rc$ are assumed to be reachable from $u_0$)}, and then generates the sequence of labels in $C_m$ indefinitely.
Let $\pi_1$ and $\pi_2$ be policies that select $a_1$ and $a_2$ in $s_0$, respectively. Then, we have
\begin{align*}
    J^{\M}_{\Rc}(\pi_1) \geq p\cdot\Rc_{\avg}(C_{+})
    \geq \Rc_{\avg}(C_m) + \frac{\varepsilon}{4}
    = J^{\M}_{\Rc}(\pi_2) + \frac{\varepsilon}{4}.
\end{align*}
However $J^{\M}_{\p}(\pi_1) = p < 1 = J^{\M}_{\p}(\pi_2)$, which is a contradiction.\qed
\end{proof}

Note that the above theorem claims only the non-existence of \emph{abstract} reward machines for the LTL specification $\Box\Diamond b$, whereas Theorem~\ref{thm:discount} holds for arbitrary reward machines and history dependent reward functions. We do not rule out the possibility of a specification translation that constructs different ARMs (with limit-average rewards) for the same LTL objective depending on $S$, $A$, $s_0$ and $L$. This leads to the following natural question.

\begin{open}
Does there exist an optimality preserving specification translation from LTL specifications to reward machines with limit-average rewards?
\end{open}

\subsection{Sampling-based Reduction}

The previous section suggests that keeping the MDP $\M$ fixed might be insufficient for reducing LTL specifications to reward-based ones. In this section, we formalize the notion of a \emph{sampling-based reduction} where we are allowed to modify the MDP $\M$ in a way that makes it possible to simulate the modified MDP $\bar{\M}$ using a simulator for $\M$ without the knowledge of the transition probabilities of $\M$.

Given an RL task $(\M, \p)$ we want to construct another RL task $(\bar{\M},\p')$ and a function $f$ that maps policies in $\bar{\M}$ to policies in $\M$ such that for any policy $\bar{\pi}\in\Pi_{\opt}(\bar{\M},\p')$, we have $f(\bar{\pi}) \in \Pi_{\opt}(\M,\p)$. Since it should be possible to simulate $\bar{\M}$ without the knowledge of the transition probability function $P$ of $\M$, we impose several constraints on $\bar{\M}$.

Let $\M = (S,A,s_0, P, L)$ and $\bar{\M} = (\bar{S}, \bar{A}, \bar{s}_0, \bar{P}, \bar{L})$. First, it must be the case that $\bar{S}$, $\bar{A}$, $\bar{s}_0$, $\bar{L}$ and $f$ are independent of $P$. Second, since the simulator of $\bar{\M}$ uses the simulator of $\M$ we can assume that at any time, the state of the simulator of $\bar{\M}$ includes the state of the simulator of $\M$. Formally, there is a map $\beta:\bar{S}\to S$ such that for any $\bar{s}$, $\beta(\bar{s})$ is the state of $\M$ stored in $\bar{s}$. Since it is only possible to simulate $\M$ starting from $s_0$ we must have $\beta(\bar{s}_0) = s_0$. Next, when taking a step in $\bar{\M}$, a step in $\M$ may or may not occur, but the probability that a transition is sampled from $\M$ should be independent of $P$. Given these desired properties, we are ready to define a step-wise sampling-based reduction.

\begin{definition}\label{def:sim}
A step-wise sampling-based reduction is a computable function $\F$ that maps the tuple $(S, A, s_0, L, \p)$ to a tuple $(\bar{S}, \bar{A}, \bar{s}_0, \bar{L}, f, \beta, \alpha, q_1, q_2, \p')$ where $f: \Pi(\bar{S}, \bar{A})\to\Pi(S,A)$, $\beta:\bar{S}\to S$, $\alpha:\bar{S}\times \bar{A}\to \D(A)$, $q_1: \bar{S}\times \bar{A}\times \bar{S}\to[0,1]$, $q_2: \bar{S}\times \bar{A}\times A\times \bar{S}\to[0,1]$ and $\p'$ is a specification such that
\begin{itemize}
    \item $\beta(\bar{s}_0) = s_0$,
    \item $q_1(\bar{s}, \bar{a}, \bar{s}') = 0$ if $\beta(\bar{s})\neq\beta(\bar{s}')$ and,
    \item for any $\bar{s}\in\bar{S}$, $\bar{a}\in\bar{A}$, $a\in A$, and $s'\in S$ we have \begin{equation}\label{eq:norm}\sum_{\bar{s}'\in\beta^{-1}(s')}q_2(\bar{s},\bar{a}, a, \bar{s}') = {1-\sum_{\bar{s}'\in\bar{S}}q_1(\bar{s},\bar{a},\bar{s}')}.
    \end{equation}
\end{itemize}
For any transition probability function $P:S\times A\times S\to[0,1]$, the new transition probability function $\bar{P}:\bar{S}\times \bar{A}\times \bar{S}\to [0,1]$ is defined by
    \begin{equation}\label{eq:newprob}
        \bar{P}(\bar{s},\bar{a},\bar{s}') = q_1(\bar{s},\bar{a},\bar{s}') + \E_{a\sim \alpha(\bar{s}, \bar{a})}[q_2(\bar{s},\bar{a},a,\bar{s}')P(\beta(\bar{s}),a,\beta(\bar{s}'))].
    \end{equation}
\end{definition}

In Equation~\ref{eq:newprob}, $q_1(\bar{s}, \bar{a}, \bar{s}')$ denotes the probability with which $\bar{\M}$ steps to $\bar{s}'$ without sampling a transition from $\M$. In the event that a step in $\M$ does occur, $\alpha(\bar{s}, \bar{a})(a)$ gives the probability of the action $a$ taken in $\M$ and $q_2(\bar{s}, \bar{a}, a, \bar{s}')$ is the (unnormalized) probability with which $\bar{\M}$ transitions to $\bar{s}'$ given that action $a$ in $\M$ caused a transition to $\beta(\bar{s}')$. It is easy to see that, for any $P$, $\bar{P}$ defined in Equation~\ref{eq:newprob} is a valid transition probability function. 

\begin{lemma}\label{lem:validprob}
Given a step-wise sampling-based reduction $\F$, for any MDP $\M = (S,A,s_0,P,L)$ and specification $\p$, the function $\bar{P}$ defined by $\F$ is a valid transition probability function.
\end{lemma}

\begin{proof}
It is easy to see that $\bar{P}(\bar{s},\bar{a},\bar{s}')\geq 0$ for all $\bar{s}, \bar{s}'\in\bar{S}$ and $\bar{a}\in\bar{A}$. Now for any $\bar{s}\in\bar{S}$ and $\bar{a}\in\bar{A}$, letting $\sum_{\bar{s}'}q_1(\bar{s},\bar{a},\bar{s}') = p(\bar{s},\bar{a})$, we have
\begin{align*}
    \sum_{\bar{s}'\in\bar{S}}P(\bar{s},\bar{a}, \bar{s}') &= p(\bar{s},\bar{a}) + \sum_{\bar{s}'\in\bar{S}}\E_{a\sim \alpha(\bar{s}, \bar{a})}[q_2(\bar{s},\bar{a},a,\bar{s}')P(\beta(\bar{s}),a,\beta(\bar{s}'))]\\
    &= p(\bar{s},\bar{a}) + \E_{a\sim \alpha(\bar{s}, \bar{a})}\Big[\sum_{\bar{s}'\in\bar{S}}q_2(\bar{s},\bar{a},a,\bar{s}')P(\beta(\bar{s}),a,\beta(\bar{s}'))\Big]\\
    &= p(\bar{s},\bar{a}) + \E_{a\sim \alpha(\bar{s}, \bar{a})}\Big[\sum_{s'\in S}\sum_{\bar{s}'\in\beta^{-1}(s')}q_2(\bar{s},\bar{a},a,\bar{s}')P(\beta(\bar{s}),a,\beta(\bar{s}'))\Big]\\
    &= p(\bar{s},\bar{a}) + \E_{a\sim \alpha(\bar{s}, \bar{a})}\Big[\sum_{s'\in S}P(\beta(\bar{s}),a,s')\sum_{\bar{s}'\in\beta^{-1}(s')}q_2(\bar{s},\bar{a},a,\bar{s}'))\Big]\\
    &= p(\bar{s},\bar{a}) + \E_{a\sim \alpha(\bar{s}, \bar{a})}\Big[\sum_{s'\in S}P(\beta(\bar{s}),a,s')(1-p(\bar{s},\bar{a}))\Big]\\
    &= 1,
\end{align*}
where the penultimate step followed from Equation~\ref{eq:norm}.\qed
\end{proof}

\begin{example}\label{ex:rm}
A simple example of a step-wise sampling-based reduction is the product construction used to translate reward machines to regular reward functions \cite{icarte2018using}. Let $\Rc = (U, u_0, \delta_u, \delta_r)$. Then, we have $\bar{S} = S\times U$, $\bar{A} = A$, $\bar{s}_0 = (s_0, u_0)$, $\bar{L}(s,u) = L(s)$, $\beta(s, u) = s$, $\alpha(a)(a') = \mathbbm{1}(a'=a)$, $q_1 = 0$, and $q_2((s,u), a, a', (s', u')) = \mathbbm{1}(u' = \delta_u(u, s'))$. The specification $\p'$ is a reward function given by $R((s,u), a, (s',u')) = \delta_r(u)(s, a, s')$, and $f(\bar{\pi})$ is a policy that keeps track of the reward machine state and acts according to $\bar{\pi}$.
\end{example}

Given an MDP $\M = (S,A,s_0,P,L)$ and a specification $\p$, the reduction $\F$ defines a unique triplet $(\bar{\M},\p',f)$ with $\bar{\M}=(\bar{S}, \bar{A}, \bar{s}_0, \bar{P}, \bar{L})$, where $\bar{S}, \bar{A}, \bar{s}_0, \bar{L}, f$ and $\p'$ are obtained by applying $\F$ to $(S, A, s_0, L, \p)$ and $\bar{P}$ is defined by Equation~\ref{eq:newprob}. We let $\F(\M,\p)$ denote the triplet $(\bar{\M},\p',f)$. Given a simulator $\S$ of $\M$, we can construct a simulator $\bar{\S}$ of $\bar{\M}$ as follows.

\begin{algorithm}[tb]
\caption{Step function of the simulator $\bar{\S}$ of $\bar{\M}$ given $\beta$, $\alpha$, $q_1$, $q_2$ and a simulator $\S$ of $\M$.}
\label{alg:step_mbar}
\begin{algorithmic}
\FUNCTION{$\bar{\S}.\code{step(}\bar{a}\code{)}$}
\STATE $\bar{s}\gets\bar{\S}.\code{state}$ 
\STATE $p \gets\sum_{\bar{s}'}q_{1}(\bar{s},\bar{a}, \bar{s}')$
\STATE $x\sim\code{Uniform}(0,1)$
\IF{$x\leq p$}
\STATE $\bar{\S}.\code{state}\gets \bar{s}'\sim \dfrac{q_1(\bar{s}, \bar{a}, \bar{s}')}{p}$
\ELSE
\STATE $a\sim\alpha(\bar{s},\bar{a})$
\STATE $s'\gets\S.\code{step(}a\code{)}$
\STATE $\bar{\S}.\code{state}\gets \bar{s}'\sim \dfrac{q_2(\bar{s}, \bar{a}, a, \bar{s}')\mathbbm{1}(\beta(\bar{s}')=s')}{1 - p}$\hfill\COMMENT{Ensures $\beta(\bar{s}') = s'$}
\ENDIF
\STATE \textbf{return} $\bar{\S}.\code{state}$
\ENDFUNCTION
\end{algorithmic}
\end{algorithm}

\begin{description}
\item[\rm $\bar{\S}.\code{reset()}$:] This function internally sets the current state of the MDP to $\bar{s}_0$ and calls the reset function of $\M$.

\item[\rm $\bar{\S}.\code{step(}\bar{a}\code{)}$:] This function is outlined in Algorithm~\ref{alg:step_mbar}. We use $\bar{s}'\sim \Delta(\bar{s}')$ to denote that $\bar{s}'$ is sampled from the distribution defined by $\Delta$. It takes a step without calling $\S.\code{step}$ with probability $p$. Otherwise, it samples an action $a$ according to $\alpha(\bar{s},\bar{a})$, calls $\S.\code{step(}a\code{)}$ to get next state $s'$ of $\M$ and then samples an $\bar{s}'$ satisfying $\beta(\bar{s}')=s'$ based on $q_2$. Equation~\ref{eq:norm} ensures that $\frac{q_2}{1-p}$ defines a valid distribution over $\beta^{-1}(s')$.
\end{description}

We call the reduction step-wise since at most one transition of $\M$ can occur during a transition of $\bar{\M}$. Under this assumption, we justify the general form of $\bar{P}$. Let $\bar{s}$ and $\bar{a}$ be fixed. Let $X_{\bar{S}}$ be a random variable denoting the next state in $\bar{\M}$ and $X_{A}$ be a random variable denoting the action taking in $\M$ (it takes a dummy value $\bot\notin A$ when no step in $\M$ is taken). Then, for any $\bar{s}'\in\bar{S}$, we have
$$\Pr[X_{\bar{S}} = \bar{s}'] = \Pr[X_{\bar{S}} = \bar{s}'\land X_{A} = \bot] + \sum_{a\in A}\Pr[X_{\bar{S}} = \bar{s}'\land X_{A} = a].$$
Now, we have
\begin{align*}
&\Pr[X_{\bar{S}} = \bar{s}'\land X_{A} = a]\\  
&=\Pr[X_{A}=a]\Pr[X_{\bar{S}} = \bar{s}'\mid X_{A} = a]\\
&=\Pr[X_{A}=a]\Pr[\beta(X_{\bar{S}}) = \beta(\bar{s}')\mid X_A=a]\Pr[X_{\bar{S}} = \bar{s}'\mid X_{A} = a, \beta(X_{\bar{S}}) = \beta(\bar{s}')]\\
&= {P(\beta(\bar{s}), a, \beta(\bar{s}'))}\cdot\Pr[X_{A}=a]\Pr[X_{\bar{S}} = \bar{s}'\mid X_{A} = a, \beta(X_{\bar{S}}) = \beta(\bar{s}')].
\end{align*}
Taking $q_1(\bar{s},\bar{a},\bar{s}') = \Pr[X_{\bar{S}} = \bar{s}'\land X_{A} = \bot]$, $\alpha(\bar{s},\bar{a})(a) = \Pr[X_{A}=a]/\Pr[X_{A}\neq\bot]$, and $q_2(\bar{s},\bar{a},a,\bar{s}') = \Pr[X_{\bar{S}} = \bar{s}'\mid X_{A} = a, \beta(X_{\bar{S}}) = \beta(\bar{s}')]\cdot \Pr[X_{A}\neq\bot]$, we obtain the form of $\bar{P}$ in Definition~\ref{def:sim}. Note that Equation~\ref{eq:norm} holds since both sides evaluate to $\Pr[X_{A}\neq\bot]$.

To be precise, it is also possible to reset the MDP $\M$ to $s_0$ in the middle of a run of $\bar{\M}$. This can be modeled by taking $\alpha(\bar{s},\bar{a})$ to be a distribution over $A\times\{0,1\}$, where $(a,0)$ represents taking action $a$ in the current state $\beta(\bar{s})$ and $(a,1)$ represents taking action $a$ in $s_0$ after a reset. We would also have $q_2:\bar{S}\times\bar{A}\times A\times\{0,1\}\times\bar{S}\to [0,1]$ and furthermore $q_1(\bar{s},\bar{a}, \bar{s}')$ can be nonzero if $\beta(\bar{s}') = s_0$. For simplicity, we use Definition~\ref{def:sim} without considering resets in $\M$ during a step of $\bar{\M}$. However, the discussions in the rest of the paper apply to the general case as well. Now we define the \emph{optimality preservation} criterion for sampling-based reductions.

\begin{definition}
A step-wise sampling-based reduction $\F$ is optimality preserving if for any RL task $(\M,\p)$ letting $(\bar{\M},\p',f) = \F(\M,\p)$ we have $f(\Pi_{\opt}(\bar{\M},\p'))\subseteq\Pi_{\opt}(\M,\p)$ where $f(\Pi) = \{f(\pi)\mid\pi\in\Pi\}$ for a set of policies $\Pi$.
\end{definition}

It is easy to see that the reduction in Example~\ref{ex:rm} is optimality preserving for both discounted-sum and limit-average rewards since $J_{\p'}^{\bar{\M}}(\bar{\pi}) = J_{\p}^{\M}(f(\bar{\pi}))$ for any policy $\bar{\pi}\in\Pi(\bar{S}, \bar{A})$. Another interesting observation is that we can reduce discounted-sum rewards with multiple discount factors $\gamma:S\to]0,1[$ to the usual case with a single discount factor.

\begin{theorem}\label{thm:twodiscount}
There is an optimality preserving step-wise sampling-based reduction $\F$ such that for any $\M=(S,A,s_0,P)$ and $\p = (R,\gamma)$, where $R:S\times A\times S\to\R$ and $\gamma:S\to]0,1[$, we have $f(\M,\p) = (\bar{\M}, \p', f)$, where $\p' = (R',\gamma')$, with $R':\bar{S}\times \bar{A}\times \bar{S}\to\R$ and $\gamma'\in]0,1[$.
\end{theorem}
\begin{proof}
Let $\bar{S} = S\sqcup\{s_\bot\}$, where $s_\bot$ is a new sink state, $\bar{A} = A$, and $\bar{s}_0 = s_0$. We set $\gamma'=\gamma_{\max} = \max_{s\in S}\gamma(s)$, and define $R'$ by $R'(s,a,s') = \frac{\gamma_{\max}}{\gamma(s)}R(s,a,s')$ if $s,s'\in S$ and $0$ otherwise. We define $\bar{P}(s_{\bot},a,s_{\bot}) = 1$ for all $a\in A$. For any $s\in S$, we have $\bar{P}(s,a,s') = \frac{\gamma(s)}{\gamma_{\max}}P(s,a,s')$ if $s'\in S$ and $\bar{P}(s,a,s_{\bot}) = 1-\frac{\gamma(s)}{\gamma_{\max}}$. Intuitively, $\bar{\M}$ transitions to the sink state $s_{\bot}$ with probability $1-\frac{\gamma(s)}{\gamma_{\max}}$ from any state $s$ on taking any action $a$ which has the effect of reducing the discount factor from $\gamma_{\max}$ to $\gamma(s)$ in state $s$ since all future rewards are $0$ after transitioning to $s_\bot$. Although we explicitly defined $\bar{P}$, note that it has the general form of Equation~\ref{eq:newprob} and can be sampled from without knowing $P$. Now, we take $\p' = (R',\gamma')$, and $f(\bar{\pi})$ to be $\bar{\pi}$ restricted to $\traj_{f}(S,A)$. It is easy to see that for any $\bar{\pi}\in\Pi(\bar{S},A)$, we have $J_{\p'}^{\bar{\M}}(\bar{\pi}) = J_{\p}^{\M}(f(\bar{\pi}))$; therefore, this reduction preserves optimality.\qed
\end{proof}

\subsection{Reductions from Temporal Logic Specifications}\label{sec:existing}

\noindent A number of strategies have been recently proposed for learning policies from temporal specifications by reducing them to reward-based specifications. For instance, \cite{aksaray2016q} proposes a reduction from Signal Temporal Logic (STL) specifications to rewards in the finite horizon setting---i.e., the specification $\varphi$ is evaluated over a fixed $T_{\varphi}$-length prefix of the rollout $\zeta$.

The authors of \cite{hasanbeig2019,hasanbeig2018logically} propose a reduction from LTL specifications to discounted rewards which proceeds by first constructing a product of the MDP $\M$ with a Limit Deterministic B\"uchi automaton (LDBA) $\A_{\varphi}$ derived from the LTL formula $\varphi$ and then generates transition-based rewards in the product MDP. The strategy is to assign a fixed positive reward of $r$ when an accepting state in $\A_{\varphi}$ is reached and $0$ otherwise. As shown in \cite{moritz2019}, this strategy does not always preserve optimality if the discount factor $\gamma$ is required to be strictly less that one.
Similar approaches are proposed in \cite{yuan2019modular, spectrl, ijcai2019ltl}, though they do not provide optimality preservation guarantees.

A recent paper \cite{moritz2019} presents a step-wise sampling-based reduction from LTL specifications to limit-average rewards. It first constructs an LDBA $\A_{\varphi}$ from the LTL formula $\varphi$ and then considers a product $\M\otimes\A_{\varphi}$ of the MDP $\M$ with $\A_{\varphi}$ in which the nondeterminism of $\A_{\varphi}$ is handled by adding additional actions that represent the choice of possible transitions in $\A_{\varphi}$ that can be taken. Now, the reduced MDP $\bar{\M}$ is obtained by adding an additional sink state $\bar{s}_{\bot}$ with the property that whenever an accepting state of $\A_{\varphi}$ is reached in $\bar{\M}$, there is a $(1-\lambda)$ probability of transitioning to $\bar{s}_{\bot}$ during the next transition in $\bar{\M}$. They show that for a large enough value of $\lambda$, any policy maximizing the probability of reaching $\bar{s}_{\bot}$ in $\bar{\M}$ can be used to construct a policy that maximizes $J^{\M}_{\L_{\ltl}(\varphi)}$. As shown before, this reachability property in $\bar{\M}$ can be translated to limit-average rewards. The main drawback of this approach is that the lower bound on $\lambda$ for preserving optimality depends on the transition probability function $P$; hence, it is not possible to correctly pick the value of $\lambda$ without the  knowledge of $P$. A heuristic used in practice is to assign a default large value to $\lambda$. Their result can be summarized as follows.
\begin{theorem}[\cite{moritz2019}] There is a family of step-wise sampling-based reductions $\{\F_{\lambda}\}_{\lambda\in]0,1[}$ such that for any MDP $\M$ and LTL specification $\p = \L_{\ltl}(\varphi)$, there exists a $\lambda_{\M,\p}\in]0,1[$ such that for all $\lambda\geq \lambda_{\M,\p}$, letting $(\bar{\M}_{\lambda}, \p'_{\lambda}, f_{\lambda}) = \F_{\lambda}(\M,\p)$, we have $f_{\lambda}(\Pi_{\opt}(\bar{\M}_{\lambda},\p'_{\lambda}))\subseteq\Pi_{\opt}(\M,\p)$ and $\p'_{\lambda}=R_{\lambda}:S\times A\times S\to\R$ is a limit-average reward specification.
\end{theorem}

{The authors of \cite{atva} show that the above approach can be modified to get less sparse rewards with similar guarantees using two discount factors $\gamma_1<1$ and $\gamma_2=1$ (where $\gamma_2=1$ is only used in steps at which the reward is zero).}

Another approach \cite{bozkurt2020control} with an optimality preservation guarantee reduces LTL specifications to discounted rewards with two discount factors $\gamma_1 < \gamma_2 < 1$ which are applied in different states. This approach uses the product $\M\times\A_{\varphi}$ as $\bar{\M}$ and assigns a reward of $1-\gamma_1$ to the accepting states (where discount factor $\gamma_1$ is applied) and $0$ to the remaining states (where discount factor $\gamma_2$ is applied). Applying Theorem~\ref{thm:twodiscount} we get the following result as a corollary of the optimality preservation guarantee of this approach.

\begin{theorem}[\cite{bozkurt2020control}] There is a family of step-wise sampling-based reductions $\{\F_{\gamma}\}_{\gamma\in]0,1[}$ such that for any MDP $\M$ and LTL specification $\p = \L_{\ltl}(\varphi)$, there exists $\gamma_{\M,\p}\in]0,1[$ such that for all $\gamma\geq \gamma_{\M,\p}$, letting $(\bar{\M}_{\gamma}, \p'_{\gamma}, f_{\gamma}) = \F_{\gamma}(\M,\p)$, we have $f_{\gamma}(\Pi_{\opt}(\bar{\M}_{\gamma},\p'_{\gamma}))\subseteq\Pi_{\opt}(\M,\p)$ and $\p'_{\gamma}=(R_{\gamma}, \gamma)$ is a discounted-sum reward specification.
\end{theorem}

Similar to \cite{moritz2019}, the optimality preservation guarantee only applies to large enough $\gamma$, and the lower bound on $\gamma$ depends on the transition probability function $P$.

To the best of our knowledge, it is unknown if there exists an optimality preserving step-wise sampling-based reduction from LTL specifications to reward-based specifications that is completely independent of $P$.

\begin{open}
Does there exist an optimality preserving step-wise sampling-based reduction $\F$ such that for any RL task $(\M,\p)$ where $\p$ is an LTL specification, letting $(\bar{\M},\p',f)=\F(\M,\p)$, we have that $\p'$
is a reward-based specification (either limit-average or discounted-sum)?
\end{open}
\section{Robustness}
\label{sec:robust}

A key property of discounted reward specifications that is exploited by RL algorithms is \emph{robustness}.
In this section, we discuss the concept of robustness for specifications as well as reductions. We show that robust reductions from LTL specifications to discounted rewards are not possible due to the fact that LTL specifications are not robust.

\subsection{Robust Specifications}
\label{sec:robustspec}

A specification $\p$ is said to be \emph{robust} \citep{littman2017environment} if an optimal policy for $\p$ in an estimate $\M'$ of the MDP $\M$ achieves close to optimal performance in $\M$. Formally,
an MDP $\M = (S, A, s_0, P, L)$ is said to be {\em $\delta$-close} to another MDP ${\M}' = (S, A, {s}_0, {P}', {L})$ if their states, actions, initial states, and labeling functions are identical and their transition probabilities differ by at most a $\delta$ amount---i.e., $$|P(s,a,s') - {P}'(s,a,s')| \leq \delta$$ for all $s,s'\in S$ and $a\in A$.

\begin{definition}\label{def:robspec}
A specification $\p$ is {\em robust} if for any MDP $\M$ for which $\p$ is a valid specification and $\varepsilon>0$, there exists a $\delta_{\M, \varepsilon}>0$ such that if MDP ${\M}'$ is $\delta_{\M, \varepsilon}$-close to $\M$, then an optimal policy in ${\M}'$ is an $\varepsilon$-optimal policy in $\M$---i.e.,
$\Pi_{\opt}(\M',\p)\subseteq \Pi^{\varepsilon}_{\opt}(\M,\p).$
\end{definition}

\begin{figure}
    \centering

\begin{tikzpicture}[scale=0.14]
\tikzstyle{every node}+=[inner sep=0pt]
\draw [black] (26.6,-21.2) circle (3);
\draw (26.6,-21.2) node {$s_0$};
\draw [black] (26.6,-21.2) circle (2.4);
\draw [black] (48.3,-21.2) circle (3);
\draw (48.3,-21.2) node {$s_1$};
\draw [black] (61.6,-21.2) circle (3);
\draw (61.6,-21.2) node {$s_2$};
\draw [black] (61.6,-21.2) circle (2.4);
\draw [black] (29.142,-19.613) arc (117.24739:62.75261:18.147);
\fill [black] (45.76,-19.61) -- (45.28,-18.8) -- (44.82,-19.69);
\draw (37.45,-17.1) node [above] {$a_1\mbox{ }/\mbox{ }1-p_1$};
\draw [black] (23.92,-22.523) arc (324:36:2.25);
\draw (19.35,-21.2) node [left] {$a_1\mbox{ }/\mbox{ }p_1$};
\fill [black] (23.92,-19.88) -- (23.57,-19) -- (22.98,-19.81);
\draw [black] (29.6,-21.2) -- (45.3,-21.2);
\fill [black] (45.3,-21.2) -- (44.5,-20.7) -- (44.5,-21.7);
\draw (37.45,-21.7) node [below] {$a_2\mbox{ }/\mbox{ }p_2$};
\draw [black] (59.218,-23.021) arc (-55.79615:-124.20385:26.894);
\fill [black] (59.22,-23.02) -- (58.28,-23.06) -- (58.84,-23.88);
\draw (44.1,-28.17) node [below] {$a_2\mbox{ }/\mbox{ }1-p_2$};
\draw [black] (46.977,-18.52) arc (234:-54:2.25);
\draw (48.3,-13.95) node [above] {$A\mbox{ }/\mbox{ }1$};
\fill [black] (49.62,-18.52) -- (50.5,-18.17) -- (49.69,-17.58);
\draw [black] (60.277,-18.52) arc (234:-54:2.25);
\draw (61.6,-13.95) node [above] {$A\mbox{ }/\mbox{ }1$};
\fill [black] (62.92,-18.52) -- (63.8,-18.17) -- (62.99,-17.58);
\end{tikzpicture}

    \caption{Example showing non-robustness of $\L_{\safe}(\{b\})$.}
    \label{fig:rob}
\end{figure}
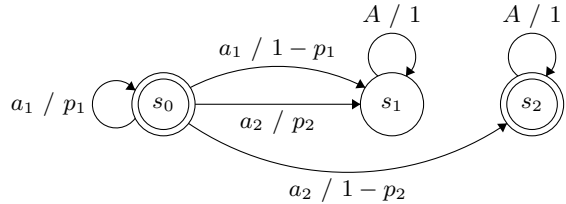
The simulation lemma in \cite{kearns2002near} proves that discounted-sum rewards are robust. On the other hand, \cite{littman2017environment} shows that language-based specifications, even safety specifications, are not robust. Here, we give a slightly modified example to show that the specification $\p=\L_{\safe}(\{b\})$ is not robust which also shows that limit-average rewards are not robust.

\begin{theorem}[\cite{littman2017environment}]\label{thm:robspec}
There exists an MDP $\M$ and a safety specification $\p$ such that, for any $\delta>0$, there is an MDP $\M_{\delta}$ that is $\delta$-close to $\M$ which satisfies $\Pi_{\opt}(\M_{\delta},\p)\cap\Pi_{\opt}^{\varepsilon}(\M,\p) = \emptyset$ for all $\varepsilon < 1$.
\end{theorem}
\begin{proof}
Consider the MDP $\M$ in Figure~\ref{fig:rob} with $p_1=p_2=1$; the double circles denote states where $b$ holds. Then, an optimal policy for $\p=\L_{\safe}(\{b\})$ always selects action $a_1$ and achieves a satisfaction probability of $1$. Now let $\M_{\delta}$ denote the same MDP with $p_1=p_2 = 1-\delta$. Then, any optimal policy for $\p$ in $\M_{\delta}$ must select $a_2$ almost surely, which is not optimal for $\M$. In fact, such a policy achieves a satisfaction probability of $0$ in $\M$. Therefore, we have $\Pi_{\opt}(\M_{\delta},\p)\cap\Pi_{\opt}^{\varepsilon}(\M,\p)=\emptyset$ for any $\delta>0$ and any $\varepsilon<1$.\qed
\end{proof}

\subsection{Robust Reductions}

In our discussion of reductions, we were interested in optimality preserving sampling-based reductions mapping an RL task $(\M,\p)$ to another task $(\bar{\M},\p')$. However, in the learning setting, if we use a PAC-MDP algorithm to compute a policy $\bar{\pi}$ for $(\bar{\M},\p')$, it might be the case that $\bar{\pi}\notin\Pi_{\opt}(\bar{\M}, \p')$. Therefore, we cannot conclude anything useful about the optimality of the corresponding policy $f(\bar{\pi})$ in $\M$ w.r.t. $\p$. Ideally, we would like to ensure that for any $\varepsilon>0$ there is a $\varepsilon'>0$ such that an $\varepsilon'$-optimal policy for ($
\bar{\M},\p')$ corresponds to an $\varepsilon$-optimal policy for ($\M,\p$).

\begin{definition}
A step-wise sampling-based reduction $\F$ is robust if for any RL task $(\M,\p)$ with $(\bar{\M},\p',f) = \F(\M,\p)$ and any $\varepsilon > 0$, there is an $\varepsilon'>0$ such that $f(\Pi_{\opt}^{\varepsilon'}(\bar{\M},\p'))\subseteq\Pi_{\opt}^{\varepsilon}({\M},\p)$.
\end{definition}

Observe that for any optimal policy $\bar{\pi}\in\Pi_{\opt}(\bar{\M},\p')$ for $\bar{\M}$ and $\p'$, we have $f(\bar{\pi})\in\bigcap_{\varepsilon>0}\Pi_{\opt}^{\varepsilon}(\M,\p) = \Pi_{\opt}(\M,\p)$; hence, a robust reduction is also optimality preserving. Although a robust reduction is preferred when translating LTL specifications to discounted-sum rewards, the following theorem shows that such a reduction is not possible. This is primarily due to the fact that LTL specifications are not robust whereas discounted-sum rewards are.

\begin{theorem}\label{thm:norobred}
Let $\prop=\{b\}$ and $\p=\L_{\safe}(\{b\})$. Then, there does not exist a robust step-wise sampling-based reduction $\F$ with the property that for any given $\M$, if $(\bar{\M},\p',f)=\F(\M,\p)$, then $\p'$ is a robust specification and $\Pi_{\opt}(\bar{\M},\p')\neq\emptyset$.
\end{theorem}
\begin{proof}
Consider the MDP $\M = (S, A, s_0, P, L)$ in Figure~\ref{fig:rob} with $p_1=p_2=1$, and consider any $\varepsilon<1$. From Theorem~\ref{thm:robspec}, we know that for any $\delta > 0$ there is an MDP $\M_{\delta} = (S, A, s_0, P_{\delta}, L)$ that is $\delta$-close to $\M$ such that $\Pi_{\opt}({\M}_{\delta},\p)\cap\Pi_{\opt}^{\varepsilon}(\M,\p)=\emptyset$. For the sake of contradiction, suppose that such a reduction exists. Then, since $\M$ and $\M_{\delta}$ represent the same input $(S, A, s_0,L, \p)$, the reduction outputs the same tuple $(\bar{S}, \bar{A}, \bar{s}_0, \bar{L}, f, \beta, \alpha, q_1, q_2, \p')$ in both cases. Furthermore, from Equation~\ref{eq:newprob} it follows, that the new transition probability functions $\bar{P}$ and $\bar{P}_{\delta}$ corresponding to $P$ and $P_{\delta}$ differ by at most a $\delta$ amount---i.e., $|\bar{P}(\bar{s},\bar{a},\bar{s}') - \bar{P}_{\delta}(\bar{s},\bar{a},\bar{s}')|\leq\delta$ for all $\bar{s}, \bar{s}'\in\bar{S}$ and $\bar{a}\in\bar{A}$. Let $\bar{\M}$ and $\bar{\M}_{\delta}$ be the MDPs corresponding to $\bar{P}$ and $\bar{P}_{\delta}$.

Let $\varepsilon'>0$ be such that $f(\Pi_{\opt}^{\varepsilon'}(\bar{\M},\p'))\subseteq\Pi_{\opt}^{\varepsilon}(\M,\p)$. Since the specification $\p'$ is robust, there is a $\delta=\delta_{\bar{\M},\varepsilon'}>0$ such that $\Pi_{\opt}(\bar{\M}_{\delta},\p')\subseteq\Pi_{\opt}^{\varepsilon'}(\bar{\M},\p')$. Let $\bar{\pi}\in\Pi_{\opt}(\bar{\M}_{\delta},\p')$ be an optimal policy for $\bar{\M}_{\delta}$ w.r.t. $\p'$. Now, since the reduction is optimality preserving, we have
$f(\bar{\pi})\in \Pi_{\opt}({\M}_{\delta},\p)$. But then, we also have $f(\bar{\pi}) \in \Pi_{\opt}^{\varepsilon}(\M,\p)$, which contradicts our assumption on $\M_{\delta}$.\qed
\end{proof}

We observe that the above result holds when the reduction is only allowed to take at most one step in $\M$ during a step in $\bar{\M}$ (and can be generalized to a bounded number of steps). This leads to the following open problem.

\begin{open}
Does there exist a robust sampling-based reduction $\F$ such that for any RL task $(\M,\p)$, where $\p$ is an LTL specification, letting $(\bar{\M},\p',f)=\F(\M,\p)$, we have that $\p'$
is a discounted reward specification (allowing $\bar{\M}$ to take unbounded number of steps in $\M$ per transition)?
\end{open}

Note that even if such a reduction is possible, simulating $\bar{\M}$ would computationally hard since there might be no bound on the time it takes for a step in $\bar{\M}$ to occur.


\label{sec:robustreduction}

\section{Reinforcement Learning from LTL Specifications}\label{sec:no_pac}
We formalized a notion of sampling-based reductions for MDPs with unknown transition probabilities. Although reducing LTL specifications to discounted rewards is a natural approach towards obtaining learning algorithms for LTL specifications, we showed that step-wise sampling-based reductions are insufficient to obtain learning algorithms with guarantees. This leads us to the natural question of whether it is possible to design learning algorithms for LTL specifications with guarantees. Unfortunately, it turns out that it is not possible to obtain PAC-MDP algorithms for safety specifications. 

\begin{theorem}\label{thm:nopac}
There does not exist a PAC-MDP algorithm for the class of safety specifications.
\end{theorem}

Theorem~\ref{thm:norobred} shows that it is not possible to obtain a PAC-MDP algorithm for safety specifications by simply applying a step-wise sampling-based reduction followed by a PAC-MDP algorithm for discounted reward specifications. Also, Theorem~\ref{thm:norobred} does not follow from Theorem~\ref{thm:nopac} because, the definition of a robust reduction allows the maximum value of $\varepsilon'$ that satisfies $f(\Pi_{\opt}^{\varepsilon'}(\bar{\M},\p'))\subseteq\Pi_{\opt}^{\varepsilon}(\bar{\M},\p)$ to depend on the transition probability function $P$ of $\M$. However the sample complexity function $h$ of a PAC-MDP algorithm (Definition~\ref{def:pac-mdp}) should be independent of $P$.

{Intuitively, Theorem~\ref{thm:nopac} follows from that fact that, when learning from simulation, it is highly likely that the learning algorithm will encounter identical transitions when the underlying MDP is modified slightly. This makes it impossible to infer an $\varepsilon$-optimal policy using a number of samples that is independent of the transition probabilities since safety specifications are not robust.}

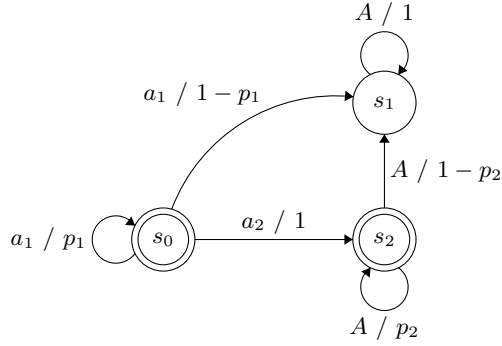
\begin{figure}
    \begin{center}
\begin{tikzpicture}[scale=0.14]
\tikzstyle{every node}+=[inner sep=0pt]

\draw [black] (20,-29.1) circle (3);
\draw (20,-29.1)  node  {$s_0$};
\draw [black] (20,-29.1) circle (2.4);
\draw [black] (41,-16.1) circle (3);
\draw (41,-16.1) node {$s_1$};
\draw [black] (41,-29.1) circle (3);
\draw (41,-29.1) node {$s_2$};
\draw [black] (41,-29.1) circle (2.4);
\draw [black] (20.8,-26.213) arc (159.34201:84.17695:16.64);
\fill [black] (38.06,-15.53) -- (37.31,-14.95) -- (37.21,-15.95);
\draw (23.72,-16.43) node [above] {$a_1\mbox{ }/\mbox{ }1-p_1$};
\draw [black] (17.32,-30.423) arc (324:36:2.25);
\draw (12.75,-29.1) node [left] {$a_1\mbox{ }/\mbox{ }p_1$};
\fill [black] (17.32,-27.78) -- (16.97,-26.9) -- (16.38,-27.71);
\draw [black] (39.677,-13.42) arc (234:-54:2.25);
\draw (41,-8.85) node [above] {$A\mbox{ }/\mbox{ }1$};
\fill [black] (42.32,-13.42) -- (43.2,-13.07) -- (42.39,-12.48);
\draw [black] (42.323,-31.78) arc (54:-234:2.25);
\draw (41,-36.35) node [below] {$A\mbox{ }/\mbox{ }p_2$};
\fill [black] (39.68,-31.78) -- (38.8,-32.13) -- (39.61,-32.72);
\draw [black] (23,-29.1) -- (38,-29.1);
\fill [black] (38,-29.1) -- (37.2,-28.6) -- (37.2,-29.6);
\draw (30.5,-28.6) node [above] {$a_2\mbox{ }/\mbox{ }1$};
\draw [black] (41,-26.1) -- (41,-19.1);
\fill [black] (41,-19.1) -- (40.5,-19.9) -- (41.5,-19.9);
\draw (41.5,-22.6) node [right] {$A\mbox{ }/\mbox{ }1-p_2$};
\end{tikzpicture}
\end{center}
    \caption{A class of MDPs for showing no PAC-MDP algorithm exists for safety specifications.}
    \label{fig:nopacex}
\end{figure}

\begin{proof}
Suppose there is a PAC-MDP algorithm $\A$ for the class of safety specifications. Consider $\prop=\{b\}$ and the family of MDPs shown in Figure~\ref{fig:nopacex} where double circles denote states at which $b$ holds. Let $\p=\L_{\safe}(\{b\})$ and $0 < \varepsilon<\frac{1}{2}$. For any $\delta>0$, we use $\M_{\delta}^1$ to denote the MDP with $p_1=1$ and $p_2=1-\delta$, and $\M_{\delta}^2$ to denote the MDP with $p_1=1-\delta$ and $p_2=1$. Finally, let $\M$ denote the MDP with $p_1=p_2=1$. Now we have the following lemma.

\begin{lemma}\label{lem:nopac}
For any $\delta\in]0,1[$, we have $\Pi_{\opt}^{\varepsilon}(\M_{\delta}^1,\p)\cap\Pi_{\opt}^{\varepsilon}(\M_{\delta}^2,\p) = \emptyset$.
\end{lemma}
\begin{proof}
Suppose $\pi\in\Pi_{\opt}^{\varepsilon}(\M^1_{\delta},\p)$ is an $\varepsilon$-optimal policy for $\M^1_{\delta}$ w.r.t. $\p$. Let $x_i = \pi((s_0a_1)^is_0)(a_1)$ denote the probability that $\pi$ chooses $a_1$ after $i$ self-loops in $s_0$. Then $J_{\p}^{\M_{\delta}^1}(\pi)=\lim_{t\to\infty}\prod_{i=0}^{t}x_i$ since choosing $a_2$ in $s_0$ leads to eventual violation of the safety specification. The policy $\pi_1^*$ that always chooses $a_1$ achieves a value of $J_{\p}^{\M^1_{\delta}}(\pi_1^*) = \J^*(\M_{\delta}^1,\p) = 1$. Since $\pi\in\Pi_{\opt}^{\varepsilon}(\M^1_{\delta},\p)$ we have $\lim_{t\to\infty}\prod_{i=0}^{t}x_i\geq 1-\varepsilon$. Therefore $\prod_{i=0}^{t}x_i\geq 1-\varepsilon$ for all $t\in\N$ since $z_t = \prod_{i=0}^{t}x_i$ is a non-increasing sequence.

Now let $E_t = \code{Cyl}((s_0a_1)^ts_1)$ denote the set of all runs that reach $s_1$ after exactly $t$ steps while staying in $s_0$ until then. We have $$\D_{\pi}^{\M^2_{\delta}}(E_t) = (1-p_1)p_1^{t-1}\prod_{i=0}^{t-1} x_i \geq \delta(1-\delta)^{t-1}(1-\varepsilon).$$ Since $\{E_t\}_{t=1}^{\infty}$ are pairwise disjoint sets, letting $E = \bigcup_{t=1}^{\infty}E_t$, we have $$\D_{\pi}^{\M^2_{\delta}}(E) = \sum_{t=1}^{\infty}\D_{\pi}^{\M^2_{\delta}}(E_t)\geq \sum_{t=1}^{\infty}\delta(1-\delta)^{t-1}(1-\varepsilon) = 1-\varepsilon.$$

But we have that $E\subseteq B = \{\zeta\in\traj(S,A)\mid L(\zeta)\notin\L_{\safe}(\{b\})\}$ and hence $J_{\p}^{\M^2_{\delta}}(\pi) = 1-\D_{\pi}^{\M^2_{\delta}}(B)\leq1-\D_{\pi}^{\M^2_{\delta}}(E)\leq\varepsilon$. Any policy ${\pi_2^*}$ that picks $a_2$ in the first step achieves $J_{\p}^{\M^2_{\delta}}(\pi^*_2) = \J^*(\M_{\delta}^2,\p) = 1$. Since $\varepsilon<\frac{1}{2}$, we have $J_{\p}^{\M^2_{\delta}}(\pi) \leq \varepsilon < \frac{1}{2}< 1-\varepsilon=\J^*(\M_{\delta}^2,\p)-\varepsilon$ which implies $\pi\notin \Pi^{\varepsilon}_{\opt}(\M_{\delta}^2,\p)$.
Therefore $\Pi_{\opt}^{\varepsilon}(\M^1_{\delta}, \p)\cap\Pi^{\varepsilon}_{\opt}(\M_{\delta}^2,\p)=\emptyset$ for all $\delta\in]0,1[$.\qed
\end{proof}
Now let $h$ be the sample complexity function of $\A$ as in Definition~\ref{def:pac-mdp}. We let $p = 0.1$ and $N = h(|S|,|A|,|\p|,\frac{1}{p},\frac{1}{\varepsilon})$. We let $K=2N+1$ and choose $\delta\in]0,1[$ such that $(1-\delta)^K\geq 0.9$. Let $\{\pi_n\}_{n=1}^{\infty}$ denote the sequence of output policies of $\A$ when run on $\M$ with the precision $\varepsilon<\frac{1}{2}$ and $p=0.1$. For $j\in\{1,2\}$, let $E_j$ denote the event that at most $N$ out of the first $K$ policies $\{\pi_n\}_{n=1}^K$ are \emph{not} $\varepsilon$-optimal for $\M_{\delta}^j$ (when $\A$ is run on $\M$). Then we have $\Pr_{\A}^{\M}(E_1) + \Pr_{\A}^{\M}(E_2) \leq 1$ because $E_1$ and $E_2$ are disjoint events (due to Lemma~\ref{lem:nopac}).

For $j\in\{1,2\}$, we let $\{\pi^{j}_n\}_{n=1}^{\infty}$ be the sequence of output policies of $\A$ when run on $\M_{\delta}^j$ with the same precision $\varepsilon$ and $p=0.1$. 
Let $F_j$ denote the event that at most $N$ out of the first $K$ policies $\{\pi^j_n\}_{n=1}^K$ are \emph{not} $\varepsilon$-optimal for $\M_{\delta}^j$ (when $\A$ is run on $\M_{\delta}^j$). Then PAC-MDP guarantee of $\A$ gives us that $\Pr_{\A}^{\M_{\delta}^j}(F_j)\geq 0.9$ for $j\in\{1,2\}$. Now let $G_j$ denote the event that the the first $K$ samples from $\M_{\delta}^j$ correspond to the deterministic transitions in $\M$---i.e., taking $a_1$ in $s_0$ leads to $s_0$ and taking any action in $s_2$ leads to $s_2$. We have that $\Pr_{\A}^{\M_{\delta}^j}(G_j) \geq (1-\delta)^K\geq 0.9$ for $j\in\{1,2\}$.

Applying union bound, we get that $\Pr_{\A}^{\M_{\delta}^j}(F_j\land G_j)\geq 0.8$ for $j\in\{1,2\}$. The probability of any execution (sequence of output policies, actions taken, resets performed and transitions observed) of $\A$ on $\M_{\delta}^j$ that satisfies the conditions of $F_j$ and $G_j$ is less than or equal to the probability of obtaining the same execution when $\A$ is run on $\M$ and furthermore such an execution also satisfies the conditions of $E_j$. Therefore, we have $\Pr_{\A}^{\M}(E_j)\geq\Pr_{\A}^{\M_{\delta}^j}(F_j\land G_j)\geq 0.8$ for $j\in\{1,2\}$. But this contradicts the fact that $\Pr_{\A}^{\M}(E_1) + \Pr_{\A}^{\M}(E_2) \leq 1$.\qed
\end{proof}



{We can also conclude that PAC-MDP algorithms do not exist for limit-average rewards since safety specifications can be encoded using limit-average rewards. Our proof of Theorem~\ref{thm:nopac} can be modified to show the result for reachability as well.

A concurrent work \cite{yang2021reinforcement} characterizes the class of LTL specifications for which PAC-MDP algorithms exist. An LTL formula $\varphi$ is \emph{finitary} if there exists a horizon $H$ such that infinite
length words sharing the same prefix of length $H$ are either all accepted or all rejected by $\varphi$. Then, their result can be summarized as follows.}

\begin{theorem}[\cite{yang2021reinforcement}]
There exists a PAC-MDP algorithm for an LTL specification $\p=\L_{\textsc{ltl}}(\varphi)$ if and only if $\varphi$ is finitary.
\end{theorem}

Next, to the best of our knowledge, it is unknown if there is a learning algorithm that converges in the limit for the class of LTL specifications.

\begin{open}
Does there exist a learning algorithm that converges in the limit for the class of LTL specifications?
\end{open}

{Observe that algorithms that converge in the limit do not necessarily have a bound on the number of samples needed to learn an $\varepsilon$-optimal policy; instead, they only guarantee that the values of the policies $\{J_{\p}^{\M}(\pi_n)\}_{n=1}^{\infty}$ converge to the optimal value $\J^*(\M,\p)$ almost surely. Therefore, the rate of convergence can be arbitrarily small and can depend on the transition probability function $P$.}

\subsection{Exisiting PAC Results}\label{sec:pac_result}
It has been shown that one can obtain PAC algorithms for learning from logical specifications under some additional assumptions. For instance, some stochastic model checking (SMC) algorithms~\cite{daca2017faster, ashok2019pac} have PAC guarantees with the sample complexity function $h$ depending on the smallest positive probability $p_{\min}$ of $\M$. A recent paper \citep{fu2014probably} proposes a PAC algorithm for LTL specifications under the assumption that the structure of the MDP $\M$ (transitions with non-zero probability) is known.
\section{Concluding Remarks}

We have established a formal framework for sampling-based reductions of RL tasks. Given an RL task (an MDP and a specification), the goal is to generate another RL task such that the transformation preserves optimal solutions and is (optionally) robust. A key challenge is that the transformation must be defined without the knowledge of the transition probabilities. 

    
    
    

{Our framework offers a unified view of the literature on RL from logical specifications, in which an RL task with a logical specification is transformed to one with a reward-based specification. We define optimality preserving as well as robust sampling-based reductions for RL tasks. 
Specification translations are special forms of sampling-based reductions in which  the underlying MDP is not altered. We show that specification translations from LTL to reward machines with discounted-sum objectives do not preserve optimal solutions. This motivates the need for transformations in which the underlying MDP may be altered. 
By revisiting such transformations from existing literature within our framework, we expose the nuances in their theoretical guarantees about optimality preservation.
Specifically, known transformations from LTL specifications to rewards are not strictly optimality preserving sampling-based reductions since they depend on parameters which are not available in the RL setting such as some information about the transition probabilities of the MDP. 
We show that LTL specifications, which are non-robust,  cannot be robustly transformed to robust specifications, such as discounted-sum rewards. We wrap up by proving that there are LTL specifications that do not admit PAC-MDP learning algorithms. 

Finally, we are left with multiple open problems. Notably, it is unknown whether there exists a learning algorithm  for LTL that converges in the limit and does not depend on any unavailable information about the MDP. However, existing algorithms for learning from LTL specifications have been demonstrated to be effective in practice, even for continuous state MDPs.
This shows that there is a gap between the theory and practice suggesting that we need better measures for theoretical analysis of such algorithms; for instance, realistic MDPs may have additional structure that makes learning possible.}\\

\noindent\textbf{Acknowledgements.} We would like to thank Michael Littman, Sheila McIlraith, Ufuk Topcu, Ashutosh Trivedi and the anonymous reviewers for their feedback on an early version of this paper. This work was supported in part by CRA/NSF Computing Innovations Fellow
Award, DARPA
Assured Autonomy project under Contract No. FA8750-18-C-0090, NSF Awards CCF-1910769 and CCF-1917852, ARO Award W911NF-20-1-0080 and ONR
award N00014-20-1-2115.

%
%
\bibliographystyle{splncs04}
\bibliography{main}

\begin{thebibliography}{10}
\providecommand{\url}[1]{\texttt{#1}}
\providecommand{\urlprefix}{URL }
\providecommand{\doi}[1]{https://doi.org/#1}

\bibitem{abel2021expressivity}
Abel, D., Dabney, W., Harutyunyan, A., Ho, M.K., Littman, M., Precup, D.,
  Singh, S.: On the expressivity of markov reward. Advances in Neural
  Information Processing Systems  \textbf{34} (2021)

\bibitem{abounadi2001learning}
Abounadi, J., Bertsekas, D., Borkar, V.S.: Learning algorithms for {M}arkov
  decision processes with average cost. SIAM Journal on Control and
  Optimization  \textbf{40}(3),  681--698 (2001)

\bibitem{aksaray2016q}
Aksaray, D., Jones, A., Kong, Z., Schwager, M., Belta, C.: Q-learning for
  robust satisfaction of signal temporal logic specifications. In: Conference
  on Decision and Control (CDC). pp. 6565--6570. IEEE (2016)

\bibitem{ashok2019pac}
Ashok, P., K{\v{r}}et{\'\i}nsk{\`y}, J., Weininger, M.: Pac statistical model
  checking for markov decision processes and stochastic games. In:
  International Conference on Computer Aided Verification. pp. 497--519.
  Springer (2019)

\bibitem{BaierAFK18}
Baier, C., de~Alfaro, L., Forejt, V., Kwiatkowska, M.: Model checking
  probabilistic systems. In: Handbook of Model Checking, pp. 963--999. Springer
  (2018)

\bibitem{bozkurt2020control}
Bozkurt, A.K., Wang, Y., Zavlanos, M.M., Pajic, M.: Control synthesis from
  linear temporal logic specifications using model-free reinforcement learning.
  In: 2020 IEEE International Conference on Robotics and Automation (ICRA). pp.
  10349--10355. IEEE (2020)

\bibitem{brafman2018ltlf}
Brafman, R., De~Giacomo, G., Patrizi, F.: Ltlf/ldlf non-markovian rewards. In:
  Proceedings of the AAAI Conference on Artificial Intelligence. vol.~32 (2018)

\bibitem{ijcai2019ltl}
Camacho, A., Toro~Icarte, R., Klassen, T.Q., Valenzano, R., McIlraith, S.A.:
  {LTL} and beyond: Formal languages for reward function specification in
  reinforcement learning. In: International Joint Conference on Artificial
  Intelligence. pp. 6065--6073 (7 2019)

\bibitem{daca2017faster}
Daca, P., Henzinger, T.A., K{\v{r}}et{\'\i}nsk{\`y}, J., Petrov, T.: Faster
  statistical model checking for unbounded temporal properties. ACM
  Transactions on Computational Logic (TOCL)  \textbf{18}(2),  1--25 (2017)

\bibitem{de2019foundations}
De~Giacomo, G., Iocchi, L., Favorito, M., Patrizi, F.: Foundations for
  restraining bolts: Reinforcement learning with ltlf/ldlf restraining
  specifications. In: Proceedings of the International Conference on Automated
  Planning and Scheduling. vol.~29, pp. 128--136 (2019)

\bibitem{fu2014probably}
Fu, J., Topcu, U.: Probably approximately correct {MDP} learning and control
  with temporal logic constraints. In: Robotics: Science and Systems (2014)

\bibitem{moritz2019}
Hahn, E.M., Perez, M., Schewe, S., Somenzi, F., Trivedi, A., Wojtczak, D.:
  Omega-regular objectives in model-free reinforcement learning. In: Tools and
  Algorithms for the Construction and Analysis of Systems. pp. 395--412 (2019)

\bibitem{atva}
Hahn, E.M., Perez, M., Schewe, S., Somenzi, F., Trivedi, A., Wojtczak, D.:
  Faithful and effective reward schemes for model-free reinforcement learning
  of omega-regular objectives. In: Hung, D.V., Sokolsky, O. (eds.) Automated
  Technology for Verification and Analysis. pp. 108--124. Springer
  International Publishing, Cham (2020)

\bibitem{hahn2020model}
Hahn, E.M., Perez, M., Schewe, S., Somenzi, F., Trivedi, A., Wojtczak, D.:
  Model-free reinforcement learning for stochastic parity games. In: 31st
  International Conference on Concurrency Theory (CONCUR 2020). Schloss
  Dagstuhl-Leibniz-Zentrum f{\"u}r Informatik (2020)

\bibitem{hahn2021model}
Hahn, E.M., Perez, M., Schewe, S., Somenzi, F., Trivedi, A., Wojtczak, D.:
  Model-free reinforcement learning for lexicographic omega-regular objectives.
  In: International Symposium on Formal Methods. pp. 142--159. Springer (2021)

\bibitem{hasanbeig2019}
Hasanbeig, M., Kantaros, Y., Abate, A., Kroening, D., Pappas, G.J., Lee, I.:
  Reinforcement learning for temporal logic control synthesis with
  probabilistic satisfaction guarantees. In: Conference on Decision and Control
  (CDC). pp. 5338--5343 (2019)

\bibitem{hasanbeig2018logically}
Hasanbeig, M., Abate, A., Kroening, D.: Logically-constrained reinforcement
  learning. arXiv preprint arXiv:1801.08099  (2018)

\bibitem{icarte2018using}
Icarte, R.T., Klassen, T., Valenzano, R., McIlraith, S.: Using reward machines
  for high-level task specification and decomposition in reinforcement
  learning. In: International Conference on Machine Learning. pp. 2107--2116.
  PMLR (2018)

\bibitem{icarte2020reward}
Icarte, R.T., Klassen, T.Q., Valenzano, R., McIlraith, S.A.: Reward machines:
  Exploiting reward function structure in reinforcement learning. arXiv
  preprint arXiv:2010.03950  (2020)

\bibitem{jiang2020temporallogicbased}
Jiang, Y., Bharadwaj, S., Wu, B., Shah, R., Topcu, U., Stone, P.:
  Temporal-logic-based reward shaping for continuing learning tasks (2020)

\bibitem{spectrl}
Jothimurugan, K., Alur, R., Bastani, O.: A composable specification language
  for reinforcement learning tasks. In: Advances in Neural Information
  Processing Systems. vol.~32, pp. 13041--13051 (2019)

\bibitem{jothimurugan2021compositional}
Jothimurugan, K., Bansal, S., Bastani, O., Alur, R.: Compositional
  reinforcement learning from logical specifications. In: Advances in Neural
  Information Processing Systems (2021)

\bibitem{jothimurugan2022specification}
Jothimurugan, K., Bansal, S., Bastani, O., Alur, R.: Specification-guided
  learning of nash equilibria with high social welfare  (2022)

\bibitem{kakade2003sample}
Kakade, S.M.: On the sample complexity of reinforcement learning. University of
  London, University College London (United Kingdom) (2003)

\bibitem{kearns2002near}
Kearns, M., Singh, S.: Near-optimal reinforcement learning in polynomial time.
  Machine learning  \textbf{49}(2),  209--232 (2002)

\bibitem{li2017reinforcement}
Li, X., Vasile, C.I., Belta, C.: Reinforcement learning with temporal logic
  rewards. In: IEEE/RSJ International Conference on Intelligent Robots and
  Systems (IROS). pp. 3834--3839. IEEE (2017)

\bibitem{littman2017environmentindependent}
Littman, M.L., Topcu, U., Fu, J., Isbell, C., Wen, M., MacGlashan, J.:
  Environment-independent task specifications via {GLTL} (2017)

\bibitem{littman2017environment}
Littman, M.L., Topcu, U., Fu, J., Isbell, C., Wen, M., MacGlashan, J.:
  Environment-independent task specifications via {GLTL}. arXiv preprint
  arXiv:1704.04341  (2017)

\bibitem{pnueli1977temporal}
Pnueli, A.: The temporal logic of programs. In: 18th Annual Symposium on
  Foundations of Computer Science. pp. 46--57. IEEE (1977)

\bibitem{sistla1985complexity}
Sistla, A.P., Clarke, E.M.: The complexity of propositional linear temporal
  logics. Journal of the ACM (JACM)  \textbf{32}(3),  733--749 (1985)

\bibitem{strehl2006pac}
Strehl, A.L., Li, L., Wiewiora, E., Langford, J., Littman, M.L.: {PAC}
  model-free reinforcement learning. In: Proceedings of the 23rd international
  conference on Machine learning. pp. 881--888 (2006)

\bibitem{watkins1992q}
Watkins, C.J., Dayan, P.: Q-learning. Machine learning  \textbf{8}(3-4),
  279--292 (1992)

\bibitem{ijcai2019-0557}
Xu, Z., Topcu, U.: Transfer of temporal logic formulas in reinforcement
  learning. In: International Joint Conference on Artificial Intelligence. pp.
  4010--4018 (7 2019)

\bibitem{yang2021reinforcement}
Yang, C., Littman, M., Carbin, M.: Reinforcement learning for general ltl
  objectives is intractable. arXiv preprint arXiv:2111.12679  (2021)

\bibitem{yuan2019modular}
Yuan, L.Z., Hasanbeig, M., Abate, A., Kroening, D.: Modular deep reinforcement
  learning with temporal logic specifications. arXiv preprint arXiv:1909.11591
  (2019)

\end{thebibliography}

\end{document}